\documentclass[nonacm,sigconf]{acmart}

\usepackage[a-2b]{pdfx}
\usepackage{balance}  
\usepackage{graphicx}
\usepackage[position=bottom]{subcaption}
\usepackage[ruled,linesnumbered,vlined,norelsize]{algorithm2e}
\usepackage[noend]{algpseudocode}
\usepackage{amsmath}
\usepackage{mathtools}
\usepackage{xcolor}
\usepackage[shortlabels,inline]{enumitem}
\usepackage{xspace}
\usepackage{setspace}
\usepackage{etoolbox}
\usepackage{lipsum}
\usepackage{array}
\usepackage{float}
\usepackage{listings}
\usepackage{tikz}
\usepackage{multirow}
\usepackage{pgfplots}
\pgfplotsset{compat=1.13}
\usepackage{macros}
\usepackage{tabto}

\makeatletter
\patchcmd{\@algocf@start}%
  {-1.5em}%
  {0pt}%
  {}{}%
\makeatother

\AtBeginDocument{%
  \providecommand\BibTeX{{%
    \normalfont B\kern-0.5em{\scshape i\kern-0.25em b}\kern-0.8em\TeX}}}

\setcopyright{acmcopyright}
\copyrightyear{2018}
\acmYear{2018}
\acmDOI{10.1145/1122445.1122456}

\acmConference[Woodstock '18]{Woodstock '18: ACM Symposium on Neural
  Gaze Detection}{June 03--05, 2018}{Woodstock, NY}
\acmBooktitle{Woodstock '18: ACM Symposium on Neural Gaze Detection,
  June 03--05, 2018, Woodstock, NY}
\acmPrice{15.00}
\acmISBN{978-1-4503-XXXX-X/18/06}
\acmDOI{10.1145/3517804.3524153}

\newtheorem{theorem}{Theorem}[section]
\newtheorem{lemma}[theorem]{Lemma}

\newtheorem{corollary}[theorem]{Corollary}
\newtheorem{definition}[theorem]{Definition}

\newcommand{\nop}[1]{}

\copyrightyear{2022}
\acmYear{2022}
\setcopyright{acmlicensed}
\acmConference[PODS '22] {Proceedings of the 41st ACM SIGMOD-SIGACT-SIGAI Symposium on Principles of Database Systems}{June 12--17, 2022}{Philadelphia, PA, USA.}
\acmBooktitle{Proceedings of the 41st ACM SIGMOD-SIGACT-SIGAI Symposium on Principles of Database Systems (PODS '22), June 12--17, 2022, Philadelphia, PA, USA}
\acmPrice{15.00}
\acmISBN{978-1-4503-9260-0/22/06} 

\title{Fast Parallel Hypertree Decompositions in Logarithmic Recursion Depth}

\author{Georg Gottlob}
\orcid{0000-0002-2353-5230}
\author{Matthias Lanzinger}
\orcid{0000-0002-7601-3727}
\affiliation{%
  \institution{University of Oxford}
  \city{Oxford}
  \country{UK}
}
\email{georg.gottlob@cs.ox.ac.uk}
\email{matthias.lanzinger@cs.ox.ac.uk}

\author{Cem Okulmus}
\orcid{0000-0002-7742-0439}
\author{Reinhard Pichler}
\orcid{0000-0002-1760-122X}
\affiliation{%
  \institution{TU Wien}
  \city{Vienna}
  \country{Austria}
}

\email{cokulmus@dbai.tuwien.ac.at}
\email{pichler@dbai.tuwien.ac.at}

\begin{abstract}
Various classic reasoning problems with natural hypergraph representations are known to be tractable 
when a hypertree decomposition (HD) of low width exists. The resulting algorithms are attractive 
for practical use in fields like databases and constraint satisfaction.
However, algorithmic use of HDs relies on the difficult task of 
first computing a decomposition of the hypergraph underlying a given problem instance, 
which is then used to guide the algorithm 
for this particular instance. 
The performance of 
purely sequential methods for computing HDs is inherently limited, yet the problem is, theoretically, 
amenable to parallelisation.

In this paper we propose the first algorithm for computing hypertree decompositions that is well-suited 
for parallelisation. 
The newly proposed algorithm \logk requires only a logarithmic number of recursion levels and additionally 
allows for highly parallelised 
pruning of the search space by restriction to so-called balanced separators. 
We provide a detailed experimental evaluation over the HyperBench benchmark and demonstrate that \logk outperforms the current state-of-the-art significantly.
\end{abstract}



\begin{document}

\maketitle

\section{Introduction}
\label{sect:Introduction}

Hypertree decompositions (HDs)~\cite{DBLP:journals/jcss/GottlobLS02} have
been demonstrated to be a valuable tool in a wide field of algorithmic applications. By way of structural decomposition of the hypergraph representation of problem instances, they induce tractable fragments for fundamental reasoning problems such as conjunctive query evaluation~\cite{DBLP:journals/jcss/GottlobLS02}, constraint satisfaction problems~\cite{DBLP:journals/ai/GottlobLS00}, and related counting problems~\cite{DBLP:journals/jcss/PichlerS13}. Other applications can be found in game theory, where problems such as determining Nash Equilibria~\cite{DBLP:journals/jair/GottlobGS05} and combinatorial auctions~\cite{DBLP:journals/jacm/GottlobG13} also 
become tractable in cases where HDs of bounded width exist.

In many of the listed cases, we do not only have theoretical tractability results but in fact know of algorithms which are suitable 
for practical applications. For example, in conjunctive query evaluation, HDs can be used for efficient reduction to an acyclic instance, 
which allows for linear-time solving using Yannakakis' algorithm~\cite{DBLP:conf/vldb/Yannakakis81}. Beyond practical algorithms, many of the listed problems are in fact known to be contained in the complexity class \nctwo~\cite{DBLP:journals/iandc/Cook85} if a bounded width HD exists~\cite{DBLP:journals/jacm/GottlobLS01,DBLP:journals/tcs/GottlobLS02,DBLP:conf/icdt/AfratiJRSU17}. Importantly, problems in \nctwo are considered to be highly parallelisable~\cite{DBLP:journals/iandc/Cook85} and thus the use of HDs in these areas can be even more attractive in parallelised and distributed scenarios.
The promising theoretical properties of hypertree decompositions have also been experimentally verified.
Implementations in specialised database systems have demonstrated 
the applicability of HDs in query evaluation by using them (and closely related variants), especially on difficult instances where current heuristic-based systems struggle~\cite{DBLP:conf/icde/GhionnaGGS07,DBLP:conf/cikm/GhionnaGS11,DBLP:journals/tods/AbergerLTNOR17}.

Despite these desirable properties and a demand for worst-case guarantees in various potential fields of application, the adoption of hypertree decompositions in practice has been slow. One crucial challenge that is limiting their more widespread use is the computational difficulty of constructing good HDs. In general, finding 
an HD for a given hypergraph $H$ and width at most $k$ 
is \np-hard and \wone-hard when parametrised by $k$~\cite{DBLP:journals/jcss/GottlobLS02}, but is tractable when $k$ is fixed, i.e., the problem is in \xp in the terminology of parametrised complexity. In fact, a significantly stronger upper bound  can be given. Finding 
an HD of fixed width is in the complexity class \logcfl~\cite{DBLP:journals/jcss/GottlobLS02} (contained in \nctwo) and therefore, in theory, highly parallelisable~\cite{DBLP:journals/iandc/Cook85}.
However, the theoretical parallelisability of the problem is demonstrated by construction of an appropriate Alternating Turing Machine~\cite{DBLP:conf/focs/ChandraS76} and no practical algorithm that allows for effective parallel computation of HDs is known. Here, we propose the first such algorithm.

\paragraph{Related Work}
HD computation 
has received significant attention recently. This is witnessed, for instance, 
by the development of the 
large benchmark data set HyperBench~\cite{10.1145/3440015}, novel
algorithmic approaches~\cite{DBLP:conf/ijcai/GottlobOP20,DBLP:conf/cp/FichteHLS18,10.1145/3440015},
and being the subject of a recent PACE competition~\cite{DBLP:conf/iwpec/DzulfikarFH19}. 
Moreover, a number of new theoretical results~\cite{DBLP:conf/mfcs/GottlobLPR20,JACMForthcoming} have been presented,  which 
have deepened our
understanding of the problem. Still, the development of a parallel algorithm for hypertree decomposition remains a critical open question.

The two state-of-the-art approaches for computing HDs, \detk~\cite{DBLP:journals/jea/GottlobS08} and \hdsmt~\cite{DBLP:conf/ijcai/SchidlerS21}, both rely on techniques that are inherently unsuitable for parallelisation. \detk is heavily reliant on extensive caching and would
therefore require excessive coordination between threads. In \hdsmt, the problem is encoded as an SMT instance and is therefore
limited by the lack of parallelisation strategies for SMT solvers.
While both algorithms perform well on current benchmarks, their lack of parallelisation ultimately limits them when it comes to solving large instances, i.e., finding HDs of large hypergraphs.
This situation is especially disappointing as, on the one hand,
 single-core performance apparently does not suffice to solve larger instances and, on the other hand, the problem is in fact highly parallelisable in theory.

Interestingly, in~\cite{DBLP:conf/ijcai/GottlobOP20}, a parallel algorithm \balgo is proposed for a slightly more general problem of computing  \emph{generalised} hypertree decompositions (GHDs)~\cite{DBLP:journals/jacm/GottlobMS09}.  In HDs, the so-called \emph{special condition} enforces certain constraints on the parent/ child nodes in the decomposition tree and the tree must therefore be treated as rooted. Crucially, this constraint is no longer enforced in GHDs and it is therefore also no longer necessary to consider the decomposition tree to be rooted. This additional degree of freedom is a key factor in the design of \balgo, where it ultimately allows for simple reassembly of individual decompositions of subproblems into a GHD of the full hypergraph.
However, this freedom comes at a
significant additional computational cost as the corresponding
decision problem for computing GHDs is \np-hard even for constant width 2~\cite{JACMForthcoming,DBLP:journals/jacm/GottlobMS09} (i.e., it is not even in \xp in the parametrised setting). 
In practice, this leads to an additional exponential factor in the algorithms' complexity in contrast to the complexity of algorithms for computing HDs.
In summary, current
approaches either are not amenable to effective parallelisation or
compute GHDs and therefore potentially cause exponential additional cost.
The goal of this paper is to bridge this gap and develop a parallel algorithm for computing hypertree decompositions.

\paragraph{Our Contributions}
As argued above, this goal is not achievable by a straightforward extension of
current approaches. The two principal algorithms for HDs are
inherently unsuited for parallelisation while the parallel algorithm
for GHDs fundamentally relies on the fact that GHDs are
unrooted. We therefore develop a new theoretical machinery which will allow us to construct HDs in an arbitrary order instead of being 
limited to a strict top-down or bottom-up construction of the HD. 
This machinery then allows us to build on some of the ideas of \balgo while avoiding the complexity of GHDs. Experimental evaluation demonstrates that the resulting algorithm combines the best of both worlds by scaling effectively 
with an increase of parallel threads while avoiding the exponential overhead of GHD computation.
Our main contributions are as follows:
\begin{itemize}
\item We develop a new theoretical framework of extended hypergraphs and their balanced separation, and we
show that extended hypergraphs always have a balanced separator. 
To actually find 
such balanced separators, it is crucial to apply a novel approach that 
determines pairs of parent and child nodes of an HD (rather than a single node) at a time. 
\item Based on these new results we propose a novel algorithm \logk, 
which  searches for  balanced separators at arbitrary positions in a potential HD. We argue that our algorithm is well-suited for parallelisation, 
in particular we prove a logarithmic upper bound on the recursion depth.
\item We identify a number of further optimisations of 
our basic algorithm and we incorporate them into  a parallelised reference implementation of \logk.
\item We compare the performance of \logk to \detk and \hdsmt through experiments over the HyperBench
benchmark~\cite{10.1145/3440015}. We observe that \logk outperforms the state-of-the-art 
significantly. Furthermore, we experimentally verify the parallel scaling behaviour of \logk.
\end{itemize}

\paragraph{Structure}
We formally introduce important concepts and notation in Section~\ref{sect:Preliminaries}. 
The theoretical framework of extended hypergraphs and their balanced separation is established in Section~\ref{sect:basics}. 
Building on this framework, we introduce the core ideas of the \logk algorithm and establish a logarithmic bound on its recursion depth in 
Section~\ref{sect:Algorithm:ShortForm}. The results of our empirical evaluation are presented in Section~\ref{sect:empirical}.
We conclude with Section~\ref{sect:Conclusion}.

Full details for the proposed optimisations are presented in Appendix~\ref{sect:improvements}. %

\section{Preliminaries}
\label{sect:Preliminaries}

\paragraph{ CQs, CSPs, and hypergraphs}
A \emph{hypergraph} $H = (V(H), E(H))$ is a pair consisting of a set of vertices $V(H)$ and a 
set of non-empty (hyper)edges $E(H) \subseteq 2^{V(H)}$.
We may assume w.l.o.g.\  that there are no isolated vertices, i.e., for each $v \in V(H)$, 
there is at least one edge $e \in E(H)$ with $v \in e$.
We can thus identify a hypergraph $H$ with its set of edges $E(H)$ with the understanding that 
$V(H) = \{ v \in e \mid e \in E(H) \} $.  A \emph{subhypergraph} $H'$ of $H$ 
is then simply a  subset of (the edges of) $H$. By slight abuse of notation we may thus 
write $H' \subseteq H$ with the understanding that $E(H') \subseteq E(H)$ and, hence, 
implicitly also $V(H') \subseteq V(H)$.
We are frequently dealing with sets of sets of vertices (e.g., sets of edges).
For $S \subseteq 2^{V(H)}$, we write 
$\bigcup S$ as a short-hand for the union of
such a set of sets, i.e., for $S = \{s_1, \dots, s_\ell\}$, we have $\bigcup S = \bigcup_
{i=1}^\ell s_i$.

\emph{Conjunctive Queries} (CQs), 
are arguably one of the most fundamental types of queries in the database world.
Similarly,  \emph{Constraint Satisfaction Problems} (CSPs) are among the most fundamental formalisms in Artificial Intelligence and
for modelling combinatorial problems. Formally, both are given by a first-order formula $\phi$ 
using only the connectives in $\{\exists, \land\}$ and 
disallowing $\{\forall,\vee,\neg\}$.  
Given such a formula $\phi$, 
the hypergraph $H_\phi$ corresponding to $\phi$ is defined as follows: 
$V(H_\phi) = \vars(\phi)$, i.e., the variables occurring in $\phi$;
and $E(H_\phi) = \{ \vars(a) \mid a$ is an atom in $\phi\}$. In the sequel, we will only concentrate on hypergraphs, 
with the understanding that all results ultimately apply to CQs and CSPs.

\noindent
\paragraph{Hypertree decompositions and hypertree width}
We introduce the used notation first: given a rooted tree 
$T = \langle N(T),E(T)\rangle$ with node set $N(T)$ and edge set $E(T)$, 
we write $T_u$ to denote the subtree of $T$ rooted at $u$, where 
$u$ is a node in $N(T)$. 
Analogously, we  write $\Tup$ to denote the subtree of $T$ induced by $N(T) \setminus N(T_u)$.
Intuitively, $T_u$ is the subtree of $T$ ``below'' $u$ and including $u$, while 
$\Tup$ is the subtree of $T$ ``above'' $u$.
By slight abuse of notation, we sometimes write $u \in T$ 
instead of $u \in N(T)$ to denote that $u$ is a node in $T$.
Below, we shall introduce node-labelling functions $\chi$ and $\lambda$, which 
assign to each node $u \in T$ a set of vertices or edges, respectively,  from some hypergraph $H$,
i.e.,  $\chi(u) \subseteq V(H)$
and 
$\lambda(u) \subseteq E(H)$. For a node-labelling function
$f$ with $f \in \{\chi,\lambda\}$ and a subtree $T'$ of $T$, we  define
$f(T')$ as $f(T') = \bigcup_{u' \in T'} f(u')$. 

We are now ready to recall the definitions of  hypertree decompositions  and 
hypertree width from \cite{DBLP:journals/jcss/GottlobLS02}: 
A {\em  hypertree decomposition\/} (HD) 
$\mathcal{D}$ of a hypergraph 
$H=(V(H),E(H))$ 
is a tuple $\mathcal{D} = \langle T, \chi, \lambda \rangle$, 
such that 
$T = \langle N(T),E(T)\rangle$ is a rooted tree, 
$\chi$ and $\lambda$  are node-labelling functions with 
$\chi \colon N(T) \rightarrow 2^{V(H)}$
 and
$\lambda \colon N(T) \rightarrow 2^{E(H)}$
and
the 
following conditions hold:

\begin{enumerate}[topsep=5pt]
 \item[(1)] for each $e \in E(H)$, there exists a node $u \in N(T)$ with $e \subseteq 
\chi(u)$;
 \item[(2)] for each $v \in V(H)$, the set $\{u \in N(T) \mid v \in \chi(u)\}$ is 
connected 
in $T$;
 \item[(3)] for each $u\in N(T)$, $\chi(u) \subseteq \bigcup \lambda(u)$;
\item[(4)]  for each $u\in N(T)$, $ \chi(T_u) \cap \big(\bigcup \lambda(u) \big) \subseteq \chi(u)$.
\end{enumerate}

The {\em width\/} of an HD 
$\mathcal{D} = \langle T, \chi, \lambda \rangle$
is the maximum size of the $\lambda$-labels over all nodes $u \in T$, 
i.e., $\width(\mathcal{D}) = \max_{u\in T} |\lambda(u)|$.
Moreover, the {\em hypertree width\/} of a hypergraph $H$, denoted $\hw(H)$, 
is the minimum width over all HDs of $H$.
Condition~(2) is called the ``connectedness condition''  and condition~(4) is 
referred to as the ``special condition'' in~\cite{DBLP:journals/jcss/GottlobLS02}. 
The set $\chi(u)$ is often referred to as the  ``bag'' at node $u$ and we will also 
call it the ``$\chi$-label'' of node $u$.  Analogously, the set $\lambda(u)$ will be referred to as the ``$\lambda$-label'' of $u$. 

If we drop the special condition from the above definition then we get so-called 
generalized hypertree decompositions (GHD). The width of a GHD is again defined as the maximum 
size of the $\lambda$-labels over all nodes $u \in T$ and the 
{\em generalized hypertree width\/} of a hypergraph $H$, denoted $\ghw(H)$, 
is the minimum width over all GHDs of $H$. The problem of checking if an HD of width $\leq k$ exists and,
if so, computing a concrete HD of width $\leq k$ is known to be feasible in 
polynomial time for arbitrarily chosen but fixed $k$ \cite{DBLP:journals/jcss/GottlobLS02}. 
In contrast, for GHDs, this problem has been shown to be NP-complete even if we fix $k =2$
\cite{JACMForthcoming,DBLP:journals/jacm/GottlobMS09}.

Throughout this paper, we will be dealing with a hypergraph $H$ and a tree $T$
of an HD of $H$.
To avoid confusion, we will consequently refer to the 
elements in 
$V(H)$ as {\em vertices\/} (of the hypergraph) and to the elements in $N(T)$ as 
the {\em nodes\/}
of $T$ (of the decomposition).

\section{Extended Hypergraphs and their Balanced Separation}
\label{sect:basics}

The key idea of our algorithm is to split the task of constructing an HD into subtasks of
constructing {\em parts\/} of the HD, which will be referred to as ``HD-fragments'' 
in the sequel. These HD-fragments can later be stitched together to form an HD of a given hypergraph. This splitting into
HD-fragments is realised by choosing a node $u$ of the HD and splitting the HD into one subtree
above node $u$ and possibly several subtrees below $u$. In order to keep track of how to combine
these subtrees later on, we introduce the notion of {\em special edges\/}. Intuitively, a
special edge is the set $\chi(u)$ of vertices for some node $u$ in the HD, and it is used to
keep track of the interface between the HD-fragment ``above'' node $u$ (we will denote this part of the HD as 
$\Tup$) and the HD-fragments at subtrees below node $u$. Conversely,
for each of the subtrees $T_{u_i}$ rooted at the child nodes $u_i$ of $u$, we have to keep track of
the interface to $\chi(u)$ in the form of a set $\Con$ of vertices, which is the
intersection $\chi(T_{u_i}) \cap \chi(u)$. 

At the heart of our decomposition algorithm 
in Section \ref{sect:Algorithm:ShortForm} 
will be a recursive
function \Decomp, which takes as input a subset $E'$ of the
edges $E(H)$, a set of special edges $\Sp$, and a set of vertices $\Con$. The goal of \Decomp is
to construct a fragment of an HD, such that every edge $e \in E'$ is covered by some node $u'$ in
the HD-fragment (i.e., $e \subseteq \chi(u')$), all special edges are covered by some leaf node of
this HD-fragment (hence, these are the interfaces to the HD-fragments ``below'') and $\Con$ must be
fully contained in $\chi(r)$ of the root $r$ of this HD-fragment (hence, this is the interface to
the HD-fragment ``above''). Formally, function \Decomp{} deals with {\em extended 
subhypergraphs\/} of $H$ in the following sense.

\begin{definition}[extended subhypergraph]
Let $H$ be a hypergraph. An  {\em extended subhypergraph\/} of $H$ is a triple 
 $\langle E', \Sp, \Con\rangle$ with the following properties:

\begin{itemize}[topsep=5pt]
\item  $E'$ is a subset of the edge set $E(H)$ of $H$;
\item $\Sp$ is a set of special edges, i.e.,  $\Sp \subseteq 2^{V(H)}$; 
\item $\Con$ is a set of  vertices, i.e., $\Con \subseteq V(H)$. 
\end{itemize}
\end{definition}
We now extend several crucial definitions introduced in \cite{DBLP:journals/jcss/GottlobLS02}
for hypergraphs to extended subhypergraphs.

\begin{definition}[connectedness, components]
\label{def:connectedness-components}
Let $H$ be a hypergraph, let $U \subseteq V(H)$ be a set of vertices, and 
let $H' = \langle E',\Sp,\Con \rangle$ be an extended subhypergraph of $H$.

\begin{itemize}[topsep=5pt]
\item We define {\em $[U]$-adjacency\/} as a binary relation on $E' \cup \Sp$ such that two
 (possibly special) edges $f_1,f_2 \in E' \cup \Sp$ are {\em $[U]$-adjacent\/}, if $
 (f_1 \cap f_2) \setminus U \neq \emptyset$ holds.
\item We define {\em $[U]$-connectedness\/} as the transitive closure of the {\em $
 [U]$-adjacency\/} relation.
\item A {\em $[U]$-component\/} of $H'$   is a maximally $[U]$-connected subset $C \subseteq
 E' \cup \Sp$ .
\end{itemize}
\end{definition}

\noindent
Let $S$ be a set of edges and special edges with $U = \bigcup S$. Then 
we will also use the terms $[S]$-connectedness and $[S]$-com\-po\-nents 
as a short-hand for $[U]$-connectedness and $[U]$-components, respectively.
Observe that the set $\Con$  plays no role in the above definition of connectedness and components.
This is in contrast to our definition of hypertree decompositions (HDs) of extended subhypergraphs,
which we give next.

\begin{definition}[hypertree decomposition]
\label{def:extendedHD}
Let $H$ be a hypergraph  and let 
$H' = \langle E',\Sp,\Con \rangle$ be an extended subhypergraph of  $H$.
A  hypertree decomposition (HD) of $H'$
is a tuple 
$\langle T, \chi, \lambda \rangle$, 
such that 
$T = \langle N(T),E(T)\rangle$ is a rooted tree, 
$\chi$ and $\lambda$  are node-labelling functions 
and
the 
following conditions hold:

\begin{enumerate}[topsep=5pt]
 \item[(1)] for each $u\in N(T)$, either \\
 a)  $\lambda(u) \subseteq E(H)$ and $\chi(u) \subseteq \bigcup \lambda(u)$ or \\
 b) $\lambda(u) = \{s\}$ for some $s \in \Sp$ and $\chi(u)  = s$;
 \item[(2)] each $f \in E' \cup \Sp$ is ``covered'' by some $u \in N(T)$, i.e.:\\
 a) if $f \in E'$, then  $f \subseteq \chi(u)$; \\
 b) if $f \in \Sp$, then $\lambda(u) =\{f\}$ and, hence, $\chi(u) =f$;
\item[(3)] for each $v \in \big(\bigcup E'\big) \cup \big(\bigcup \Sp\big)$, the set $\{u \in N(T)
\mid v \in \chi(u)\}$ is connected 
in $T$;
\item[(4)]  for each $u\in N(T)$, $ \chi(T_u) \cap \big(\bigcup \lambda(u) \big)\subseteq \chi(u)$;
\item[(5)] if $\lambda(u) = \{s\}$ for some $s \in \Sp$, then $u$ is a leaf  of $T$;
\item [(6)] the root $r$ of $T$ satisfies  $\Con \subseteq \chi(r)$.
\end{enumerate}
\end{definition}

\noindent
Clearly, $H$ can also be considered as an extended subhypergraph of itself by taking the triple
$\langle E(H), \emptyset, \emptyset\rangle$. Then the HDs of the extended
subhypergraph $\langle E(H), \emptyset, \emptyset\rangle$ and the HDs of hypergraph $H$ coincide. 

In \cite{DBLP:journals/jcss/GottlobLS02}, Definition 5.1, a normal form of HDs was introduced.
Below, in Definition~\ref{def:normalform},
we
will carry the notion of normal form over to HDs of extended subhypergraphs. To this end, it is
convenient to first define the set of (possibly special) edges {\em covered for the first time\/}
by some node or by some subtree of an HD.

\begin{definition}
\label{def:cov}
Let 
$H' = \langle E',\Sp,\Con \rangle$ be an extended subhypergraph of  
some hypergraph $H$ and let 
$\mathcal{D} = \langle T, \chi, \lambda \rangle$ be an HD of $H'$. 
For a node $u \in T$, we write $\cov(u)$ to denote the set of 
edges and special edges {\em covered for the first time\/} at $u$, i.e.:
$\cov(u) = \{f \in E' \cup \Sp \mid f \subseteq \chi(u)$ and for all ancestor nodes $u'$ of $u$, 
$f \not\subseteq \chi(u')$ holds$\}$. 
For a subtree $T'$ of $T$, we define $\cov(T')  = \bigcup_{u \in T'} \cov(u)$.
\end{definition}

\begin{definition}[normal form]
\label{def:normalform}
Let 
$H' = \langle E',\Sp$,  
$\Con \rangle$ be an extended subhypergraph of  
some hypergraph $H$ and let 
$\mathcal{D} = \langle T, \chi, \lambda \rangle$ be an HD of $H'$. 
We say that $\mathcal{D}$
is in {\em normal form}, if for every node $p$ in $T$ and every child node $c$ of $p$, 
the following properties hold: 

\begin{enumerate}[topsep=5pt]
\item  There is exactly one $[\chi(p)]$-component $C_p$ of $H'$ 
such that $C_p = \cov(T_c)$;
\item there exists $f \in C_p$ with $f \subseteq \chi(c)$, where 
$ C_p$ is the $[\chi(p)]$-component satisfying Condition 1;
\item $\chi(c) = \big(\bigcup \lambda(c)\big) \cap \big(\bigcup C_p\big)$, where again 
$C_p$ is the $[\chi(p)]$-component satisfying Condition 1.
\end{enumerate}
\end{definition}

\noindent
By the connectedness condition, 
the following property holds in any HD: 
if $C'$  is a $[\chi(p)]$-component of $H'$ with 
$C' \cap \cov(T_c)  \neq \emptyset$, then 
$C' \subseteq \cov(T_c)$ must hold. That is, $\cov(T_c)$ is the union of 
{\em one or several\/}
$[\chi(p)]$-components. 
Condition 1 of the normal form requires there to be 
{\em exactly one\/}
$[\chi(p)]$-component $C_p$ of $H'$ satisfying $C_p \subseteq \cov(T_c)$.

Condition 2 intuitively requires that some ``progress'' must be made by the labelling of node $c$.
Hence, in the first place, at least one vertex from $\bigcup C_p$ not already present in $\chi
(p)$ must occur in $\chi(c)$. By the connectedness condition, this is only possible if one edge $f$ from
$C_p$ occurs in $\lambda(c)$. Hence, by the special condition (i.e., condition (4) of the
definition of HDs), $f \subseteq \chi(c)$ must hold.

Condition 3 is the only place where we deviate from the normal form in \cite
{DBLP:journals/jcss/GottlobLS02}. The purpose of Condition 3 in \cite
{DBLP:journals/jcss/GottlobLS02} is to make sure that $\chi(c)$ is uniquely determined whenever
$\lambda(c)$, $\chi(p)$, and the $[\chi(p)]$-component $C_p$ from Condition 1 are known. However,
there also would have been other choices to achieve this goal. Our Condition 3 chooses $\chi(c)$ {\em minimally\/}. That
is, to ensure the special condition, $\chi(c)$ must contain all vertices from  $\bigcup \lambda
(c)$ that occur in $\chi(T_c)$. Since all edges in $C_p$ are covered at some node in $T_c$, all
vertices from $\big(\bigcup \lambda(c)\big) \cap \big(\bigcup C_p\big)$ must occur in $\chi(c)$. On
the other hand, there is no need to add further vertices to $\chi(c)$, since vertices not occurring
in $\bigcup \cov(T_c)$ can never violate the connectedness condition at node $c$ as long as we stick to our
strategy of choosing $\chi(u)$ minimally also for all nodes $u \in T_c$. In contrast, Condition~3
in 
\cite{DBLP:journals/jcss/GottlobLS02} chooses $\chi(c)$ {\em maximally\/}. That is, also all vertices
 in $\big(\bigcup \lambda(c)\big)$ that occur in $\chi(p)$ are added to $\chi(c)$. This deviation
 from the normal form in \cite{DBLP:journals/jcss/GottlobLS02} is crucial since, in our
 construction of an HD, we will be able to derive the possible sets $C_p$ 
 as soon as we have $\lambda(p)$ and $\lambda(c)$ but we will ``know'' $\chi
 (p)$ only much later in the algorithm.

We now carry over two key results from \cite{DBLP:journals/jcss/GottlobLS02}, 
whose proofs can be easily adapted to our setting of extended subhypergraphs and are therefore omitted
here.

\begin{theorem}[cf.\ \cite{DBLP:journals/jcss/GottlobLS02}, Theorem 5.4]
\label{theo:normalform}
Let $H'$ be an extended subhypergraph of  
some hypergraph $H$ and let 
$\mathcal{D}$ be an HD of $H'$ of width $k$. 
Then there exists an HD 
$\mathcal{D}'$ of $H'$ in normal form, such that\
$\mathcal{D}'$ also has width $k$. 
\end{theorem}

\begin{lemma}[cf.\ \cite{DBLP:journals/jcss/GottlobLS02}, Lemma 5.8]
\label{lemma:chi-vs-lambda}
Let  $H'$ be an exten\-ded subhypergraph of  
some hypergraph $H$ and let 
$\mathcal{D} = \langle T, \chi$, $\lambda \rangle$ be an HD in normal form of $H'$. 
Moreover, let $p, c$  be nodes in $T$ such that $p$ is the parent of $c$
and let  $C_c \subseteq C_p$ for some 
$[\chi(p)]$-component $C_p$ of $H'$. 
Then the following equivalence holds: 
$C_c$ is a $[\chi(c)]$-component  of $H'$ if and only if 
$C_c$ is a $[\lambda(c)]$-component  of $H'$.
\end{lemma}

Note that our deviation from \cite{DBLP:journals/jcss/GottlobLS02}
in the definition of the $\chi$-label of nodes in a normal-form HD
is inessential, since the ``downward'' components in an HD are not affected
by adding or removing vertices from the parent node to the $\chi$-label 
of the child node.
However, for our purposes, we need a slightly stronger version of the above lemma: 
recall that the HD construction
in \cite{DBLP:journals/jcss/GottlobLS02}  proceeds in a strict
top-down fashion. Hence, when dealing with $\lambda(c)$, 
the bag $\chi(p)$ is already known. This is due to the fact
that, initially at the root $r$, we have $\chi(r) = \bigcup \lambda(r)$
by the special condition. And then, whenever $\lambda(c)$ is determined
and $\chi(p)$ plus a $[\chi(p)]$-component are already known, 
also $\chi(c)$ can be computed. However, in our HD algorithm, 
which ``jumps into the middle'' of the HD to be constructed, 
we only have $\lambda(p)$ (but not  $\chi(p)$) available when 
determining $\lambda(c)$.
Hence, we need to slightly extend the above lemma to the following corollary, which follows
from 
Lemma \ref{lemma:chi-vs-lambda}
by an easy induction argument over the distance from the 
root of the HD:

\begin{corollary}
\label{cor:chi-vs-lambda}
Let  $H'$ be an exten\-ded subhypergraph of  
some hypergraph $H$ and let 
$\mathcal{D} = \langle T, \chi, \lambda \rangle$ be an HD in normal form of $H'$. 
Moreover, let $p, c$  be nodes in $T$ such that $p$ is the parent of $c$
and let  $C_c \subseteq C_p$ for some 
$[\lambda(p)]$-component $C_p$ of $H'$. 
Then the following equivalence holds: 
$C_c$ is a $[\chi(c)]$-com\-ponent  of $H'$ if and only if 
$C_c$ is a $[\lambda(c)]$-component  of $H'$.
\end{corollary}

As mentioned before, in our HD construction, we ``jump into the middle'' of the HD to be 
constructed. The motivation for this deviation from a strict top-down construction is that 
we want to split the work of recursively constructing fragments of the HD into pieces 
with a guaranteed upper bound on the size. Formally, we are thus aiming at a 
{\em balanced separator\/} of the HD. This concept was
already studied in \cite{DBLP:journals/ejc/AdlerGG07}, and it was shown that 
HDs always have a balanced separator. In \cite{10.1145/3440015}, balanced
separators were used to design an algorithm for GHD computation. Below, we 
formally define balanced separators for our notion of extended subhypergraphs and we show that
in an HD, a balanced separator always exists.

\begin{definition}[balanced separators]
\label{def:balsep}
Let  $H'$ be an exten\-ded subhypergraph of  
some hypergraph $H$ and let 
$\mathcal{D} = \langle T, \chi, \lambda \rangle$ be an HD of $H'$. 
A node $u$ of $T$ is a {\em balanced separator\/}, if the following 
holds:

\begin{itemize}[topsep=5pt]
\item  for every subtree $T_{u_i}$ rooted at a child node $u_i$ of $u$, 
we have $|\cov(T_{u_i})| \leq \frac{|E'| +|\Sp|}{2}$ and 
\item  $|\cov(\Tup)| < \frac{|E'| + |\Sp|}{2}$.
\end{itemize}
\end{definition}

Intuitively, this means that none of the subtrees ``below'' $u$ covers 
more than half of the edges of $E' \cup \Sp$ and 
the subtree ``above' $u$ even covers less than half of the edges of $E' \cup \Sp$.

\begin{lemma}
\label{lem:existence-balsep}
Let  $H'$ be an exten\-ded subhypergraph of  
some hypergraph $H$ and let 
$\mathcal{D} = \langle T, \chi, \lambda \rangle$ be an HD of $H'$. 
Then there exists a balanced separator in $\mathcal{D}$.
\end{lemma}

\begin{proof}[Proof of Lemma \ref{lem:existence-balsep}]
We show that, given an arbitrary HD, we 
can always find a balanced separator as follows:
Initially, we set $u = r$ for the root node $r$ of $T$ and distinguish two cases:
if $|\cov(T_{u_i})| \leq \frac{|E'| +  |\Sp|}{2}$ holds 
for every subtree $T_{u_i}$ rooted at a child node $u_i$ of $u$, 
then $u$ is a balanced separator and we are done. 
Otherwise, there exists a child node $u_i$ of $u$ such that 
$|\cov(T_{u_i})| > \frac{|E'| + |\Sp|}{2}$ holds for the subtree $T_{u_i}$ rooted at $u_i$. 
Of course, there can exist only one such child node $u_i$. Moreover, 
by $\cov(\Tupi) \cap \cov(T_{u_i}) = \emptyset$, 
we have $|\cov(\Tupi)| < \frac{|E'| + |\Sp|}{2}$.

Now set $u = u_i$ and repeat the case distinction: 
if $|\cov(T_{u_i})| \leq \frac{|E'| + |\Sp|}{2}$ holds 
for every subtree $T_{u_i}$ rooted at a child node $u_i$ of $u$, 
then $u$ is a balanced separator and we are done. 
Otherwise, there exists a child node $u_i$ of $u$ such that 
$|\cov(T_{u_i})| > \frac{|E'| + |\Sp|}{2}$ holds for the subtree $T_{u_i}$ rooted at $u_i$.
Again, there can only be one such $u_i$. So we set $u = u_i$ and
iterate the same considerations. 
This process is guaranteed to terminate since, eventually, we will reach a leaf node of $T$.
\end{proof}

\section{The \logk Algorithm}
\label{sect:Algorithm:ShortForm}

\let\oldnl\nl%
\newcommand{\nonl}{\renewcommand{\nl}{\let\nl\oldnl}}%

\begin{algorithm}
\DontPrintSemicolon
\SetKwInOut{KwPara}{Parameter}
\SetKwInput{KwType}{Type}

\KwType{Comp=($E$: Edge set, $\Sp$: Special Edge set)}

\KwIn{$H$: Hypergraph}
\KwPara{$k$: width parameter}
\KwOut{\textbf{true} if \emph{hw} of $H$ $\leq k$, else \textbf{false}}

  \SetKwFunction{algo}{Decomp}  

  \Begin{
   $H_{comp} \coloneqq$ Comp($E$: $H$, $\Sp$: $\emptyset$) \\
  \ForEach(\Comment{\textbf{RootLoop}}){
              $\lambda_r \subseteq H $\! s.t.\,  $ 1 \leq |\lambda_r| \leq k$}{
             $comps_r \coloneqq $ $[\lambda_r ]$-components of $H_{comp}$

                    \ForEach{$y \in comps_r$} {

                    $Conn_y \coloneqq V(y) \cap \bigcup \lambda_r$ \;
                      \If {\textbf{not}(\algo{$y$, $Conn_y $ })} {

                              \textbf{continue RootLoop}  \Comment{reject this root} 
                     }
                 }
        \textbf{return true}  \;
  }

    \textbf{return false} \Comment{exhausted search space} 
  }
  \SetKwProg{myalg}{function}{}{}
    \myalg{\algo{$H'\!$: Comp, $\Con$: Vertex set}} {
     \uIf(\Comment{\textbf{Base Cases}}) {
              $|\text{$H'\!.E$}| \leq k $ \textbf{\emph{and}} $|H'\!.\Sp| = 0$  } {
                            \textbf{return true} \;
     }
     \ElseIf {$|H'\!.E| = 0 $  \textbf{\emph{and}}  $|H'\!.\Sp|  = 1$} {
    \textbf{return true} \;
     }

    \ForEach(\Comment{\textbf{ParentLoop}}){
              $\lambda_p \subseteq H $\! s.t.\,  $ 1 \leq  |\lambda_p| \leq k$}  { 

             $comps_p \coloneqq $ $[\lambda_p ]$-components of $H'$ \;

             \uIf{$\exists i $ s.t. $| comps_p[i] |  >  \frac{|H'\!|}{2}$}  
             {$comp_{\low}\coloneqq comps_p[i]$ \Comment{found child comp.} 
             }                  
             \Else { 
                \textbf{continue ParentLoop}  \;
             }
                
             \If{$V(comp_{\low}) \cap Conn \not \subseteq \bigcup \lambda_p $   }  {
                    \textbf{continue ParentLoop} \Comment{connect. check} 
             }  
                
             \ForEach(\Comment{\textbf{ChildLoop}}){
                            $\lambda_c \subseteq H $\! s.t.\,  $ 1 \leq |\lambda_c| \leq k$}  {
                            $\chi_c \coloneqq   \bigcup \lambda_c \cap V(comp_{\low}) $\; 
                            \If{$V(comp_{\low}) \cap \bigcup\lambda_p \not \subseteq  \chi_c $   }  {
                                   \textbf{continue ChildLoop} \Comment{connect. check} 
                            }  

                       $comps_c \coloneqq $  $[\chi_c]$-components of  $comp_{\low}$  \;
                        \If{$\exists i $ s.t. $| comps_c[i] |  >  \frac{|H'\!|}{2}$}    {
                           \textbf{continue ChildLoop}  \;
                          }

                        \ForEach{$x \in comps_c$} {

                           $\Con_x \coloneqq V(x) \cap \chi_c$ \;
                            \If {\textbf{not}(\algo{$x$, $\Con_x $ })} {

                                     \textbf{continue ChildLoop}  \Comment{reject child}
                            }
                        }

                       $comp_{\up} \coloneqq H'\! \setminus comp_{\low} $  \Comment{pointwise diff.}

                       $comp_{\up}.\Sp = comp_{\up}.\Sp\cup \{ \chi_c \} $\;                       

                        \If{\textbf{not}(\algo{$comp_{\up}$, $\Con$})} {
                               \textbf{continue ChildLoop} \Comment{reject child} 
                        }
                           \textbf{return true} \Comment{$hw$ of $H'\! \leq k$} 
                
             }
    }
    \textbf{return false} \Comment{exhausted search space} 
       
}

\caption{\logk}
\label{alg:parHD}
\end{algorithm}

We now describe the main ideas of 
algorithm \logk. A pseudo-code description of \logk is shown in Algorithm~
\ref{alg:parHD}.

Algorithm  \logk\ 
aims at constructing an HD in normal form according to Definition~\ref{def:normalform}
of width $\leq k$ for a given hypergraph $H$ and integer $k \geq 1$. 
The task of constructing an HD is split
into subtasks that can then be processed in parallel.  
At the heart of \logk is the 
recursive function \Decomp: it takes as input
an extended subhypergraph $H'$ of $H$ in the form of parameter $H'$ of $H$ (with 
two fields $H'.E$ and $H'.\Sp$ for the sets of edges
and special edges of $H'$, respectively) plus
parameter $\Con$ for the interface of the HD-fragment to be constructed with the parts ``above'' in the final HD. It
returns ``true'' if an HD-fragment of width $\leq k$ of $H'$ exists and
``false'' otherwise. The top-level calls to function \Decomp  (line~7)  
are from the main program of  \logk which, 
in a loop (lines 3 -- 9), 
searches for the $\lambda$-label of the root node $r$ of the desired HD of $H$. By the special condition,
we have $\chi(r) = \bigcup \lambda(r)$. Hence,
the $[\lambda(r)]$-components (computed at line 4) coincide with the 
$[\chi(r)]$-components. Function \Decomp is called (on line 7) 
for each of the extended subhypergraphs of $H$ corresponding to 
the $[\lambda(r)]$-components.

The base case of function \Decomp is  reached (lines 12 -- 15) 
when the existence of such an HD-fragment is trivial, i.e.: 
either there are at most $k$ edges and no special edges left; or there is no edge and
only one special edge left. In these cases, the desired HD-fragment simply consists
of a single node whose $\lambda$-label either consists of the $\leq k$ edges or of the single special edge, 
respectively. 

Function \Decomp  is controlled by two nested loops
(lines 16 -- 39 for the outer loop and 
lines 24 -- 39 for the inner loop), which search for  the $\lambda$-labels of 
two adjacent nodes $p$ and
$c$ of the desired HD-fragment, such that $p$ is the parent and $c$ is the child. 
The idea of determining two nodes $p$ and $c$ is that, in an HD, we can determine 
$\chi(c)$ from $\lambda(c)$ if we  know $\lambda(p)$ and the $[\lambda(p)]$-component covered
by the subtree $T_c$ rooted at $c$,  see Corollary~\ref{cor:chi-vs-lambda} and 
Definition~\ref{def:normalform}.

We want node $c$ to be a balanced separator of the extended subhypergraph $H'$. 
By Lemma~\ref{lem:existence-balsep}, a balanced
separator is guaranteed to exist. To find a balanced separator $c$, we have to make sure that node $c$ satisfies the two 
conditions of Definition~\ref{def:balsep}, i.e.: (1) all of the subtrees rooted 
at a child of $c$ cover at most half of the edges and special edges in $H'$ and (2) the subtree $\Tupc$ 
``above'' $c$ covers strictly fewer than half of the 
edges and special edges in $H'$. For the second condition, observe that 
$comp_{\low}$ (chosen at line 19) is meant to be covered precisely by 
$T_c$. Note that, w.l.o.g., we are searching for an HD in normal form. 
This is why we may assume that $T_c$ covers exactly one  
$[\lambda(p)]$-component, namely $comp_{\low}$.
Further observe that  the edges and special edges covered by $\Tupc$ and 
the set $comp_{\low}$ partition the edges and special edges in $H'$. Hence, 
checking if $comp_{\low}$ contains more than half of $H'$ (on line 18) 
is equivalent to checking condition (2), i.e., 
$\Tupc$ covers strictly fewer than half of the edges and special edges in $H'$. 
In order to check that $c$ also satisfies the first 
condition of 
Definition~\ref{def:balsep}, we have to compute all 
$[\lambda(c)]$-components inside  $comp_{\low}$ (line 28) and 
check that the size of each of them is at most half of the size of $H'$ (line 29). Again, since we are only interested in HDs in normal form, we may assume
here that each subtree rooted at a child of $c$ covers exactly one of these 
$[\lambda(c)]$-components.

If such a balanced separator $\lambda(c)$ together with the $\lambda$-label $\lambda(p)$ at its parent node
has been found, several checks have to be performed to make sure that the HD-fragment under construction 
satisfies the connectedness condition. For instance, all vertices in the intersection of $Conn$ (i.e., the interface of the 
HD-fragment currently being constructed with the remaining HD ``above'' this HD-fragment)  with 
component $C_p$ (i.e., a component ``below'' node $p$) also have to occur in $\bigcup \lambda(p)$ (line 22). 

Suppose that all these checks succeed. From $\lambda(p)$ and $\lambda(c)$, we can compute 
$\chi(c)$ according to Condition~3 of the normal form introduced in 
Definition~\ref{def:normalform} (line 25).  
In the HD $\mathcal{D}'$ to be constructed for the extended subhypergraph $H'$, 
the edges and special edges of $H'$ can be split 
into 3 disjoint categories: 

\begin{enumerate}
\item the edges and special edges covered by $\chi(c)$, 
\item the edges and special edges covered by a subtree rooted at some child node of $c$, and
\item the edges and special edges covered in the HD ``above''  $c$.
\end{enumerate}

The edges and special edges  covered by $\chi(c)$ are done and need no further consideration.
The edges and special edges  in the second and third category are taken care of by recursive calls to the 
function  \Decomp (lines 33 and 37). 
To this end, we compute all $[\chi(c)]$-components $C_1,  \dots, C_m$
(line 28).
Now suppose that 
$C_1, \dots, C_\ell$ with $1 \leq \ell \leq m$ are the $[\chi(c)]$-components inside the $[\lambda(p)]$\-component $C_p$. 
Then the function  \Decomp is called recursively for each of the $[\chi(c)]$-components 
$C_1, \dots C_\ell$ (line 33). In the call for component $C_i$, 
the interface $Conn_i$  is obtained 
simply as the intersection of the vertices in 
 $C_i$ and in $\chi(c)$ (line 32).
All of the remaining $[\chi(c)]$-components are taken care of by the HD-fragment ``above'' $c$, which we try to construct
in another recursive call of function  \Decomp (line 37). 
In this recursive call,
$\chi(c)$ is added as yet another special edge -- in addition to the 
edges and special edges in the $[\chi(c)]$-components outside $C_p$. The additional special edge in the recursive call for the 
HD-part ``above'' node $c$ and the interfaces $Conn$ defined for each of the components as the intersection against $\chi(c)$,  in the recursive calls for the HD-parts ``below'' node $c$ 
ensure that we can (provided that all recursive calls of function  \Decomp are successful) stitch together the HD-fragments of these 
recursive calls to an HD-fragment of the extended subhypergraph $H'$ of $H$.

To sum up, if all recursive calls return ``true'' then the overall result of this call to function \Decomp is successful and 
returns ``true'' (line 39) . 
If at least one of the recursive calls returns ``false'', then we
have to search for a different label $\lambda(c)$ (in the next iteration of the ``ChildLoop``).  If eventually all candidates for $\lambda(c)$ have been tried out 
and none of them was successful,
then we have to search for a different label $\lambda(p)$ of the parent node $p$ 
(in the next iteration of the ``ParentLoop``)
and restart the search for $\lambda(c)$  from scratch.
Only when also all candidates for $\lambda(p)$ have been tried out and none of them was
successful, then function \Decomp returns the overall result ``false'' (line 40).

Below, we state the crucial property of  \logk, 
which makes this approach particularly well-suited for a parallel implementation.

\begin{theorem}
\label{theo:recursion-depth}
Algorithm \logk correctly checks for 
given hypergraph $H$ and integer $k \geq 1$, 
if $\hw(H) \leq k$ holds. 
The algorithm is realised by a main program 
and the recursive function \Decomp, whose
recursion depth is
bounded logarithmically in the number of edges of $H$, 
i.e., $O(\log(|H|)$. 
\end{theorem}
\begin{proof} The size of the extended subhypergraphs in the calls of function \Decomp in the main
 program can only be bounded by the size of $H$ itself. However, in every subsequent execution
 of \Decomp for some extended subhypergraph $(H'.E,H'.\Sp, \Con)$, we always choose node $c$ as a
 balanced separator. By Lemma~\ref{lem:existence-balsep}, such a balanced separator always exists.
 The extended subhypergraphs in the recursive calls are therefore guaranteed to
 have size at most $\lceil \frac{|H'|}{2}\rceil$. Note that the rounding up is necessary because
 the new  special edge $\chi(c)$ is added in the recursive call for the HD-fragment above $c$.
 Without this special edge, this component
 is guaranteed to be strictly smaller than $ \frac{|H'|}{2}$. At any rate,
 also with the upper bound $\lceil \frac{|H'|}{2}\rceil$ on the size of the extended subhypergraphs
 of $H$ in the recursive calls and with the additional calls of \Decomp from the main program, we
 thus get an upper bound $O(\log(|H|)$ on the recursion depth. 
\end{proof}

Note that we have formulated algorithm  \logk as a decision 
procedure that decides if $\hw(H) \leq k$ holds for given $H$ and $k$. In case 
of a successful computation (i.e. return-value true) it 
is easy to assemble a concrete HD of width $\leq k$ of $H$ from the HD-fragments corresponding to the various calls of procedure \Decomp. 

We  emphasize two further important properties of 
algorithm  \logk:  First, it should be noted that
the logarithmic bound on the recursion depth does not restrict the form of the HD in any way. In particular, 
it does not imply a logarithmic bound on the depth of the HD. The bound on the recursion depth is achieved by our novel 
approach of constructing the HD by recursively ``jumping'' to a balanced separator of the HD-fragment to be constructed rather than constructing the HD in a strict top-down manner as proposed in previous approaches
\cite{DBLP:journals/jcss/GottlobLS02,DBLP:journals/jea/GottlobS08}.

Second, we stress that it is crucial in our approach that we search for appropriate 
$\lambda$-labels for {\em a pair $(p,c)$ of nodes\/}, where $p$ is the parent of $c$. 
The rationale is that we need the $\lambda$-label of the parent in order to 
determine $\chi(c)$ from $\lambda(c)$. And only when we know $\chi(c)$, we 
can be sure, which edges are indeed covered by $\chi(c)$. This knowledge is crucial 
to guarantee that all of the recursive calls of 
function \Decomp have to deal with an extended subhypergraph whose size is halved, which in turn guarantees
the logarithmic upper bound on the recursion depth. This strategy is significantly different from all 
previous approaches of decomposition algorithms. 
In \cite{10.1145/3440015,DBLP:conf/ijcai/GottlobOP20}, a parallel algorithm for 
generalised hypertree decompositions is presented. There, the problem of determining the $\chi$-label of the 
balanced separator is solved by adding a big number of subedges to the hypergraph so that one may 
assume that $\chi(u) = \bigcup \lambda(u)$ holds for every node $u$. Clearly, this addition of subedges, in general, leads to 
a substantial increase of 
the hypergraph. In \cite{DBLP:phd/ethos/Akatov10},
a preliminary attempt to parallelise the computation of HDs was made without handling pairs of nodes. 
However, in the absence of $\lambda(p)$, we cannot determine $\chi(c)$ from $\lambda(c)$. Consequently,
we do not know which edges covered by $\bigcup \lambda(c)$ are ultimately covered by $\chi(c)$. 
Hence, all the edges covered by $\bigcup \lambda(c)$ would have to be added to the recursive call of \Decomp for 
the HD-part ``above'' $c$, thus destroying the balancedness and the logarithmic upper bound on the recursion depth.

\smallskip
By Theorem \ref{theo:recursion-depth}, 
Algorithm \logk guarantees a logarithmic  bound on the 
recursion depth and thus provides a good basis for 
a parallel implementation. Nevertheless it 
still leaves room for several improvements. 
For instance, we can define also negative  base cases to detect the overall answer ``false''  faster,  we can restrict the 
edges that may possibly be used in the $\lambda$-labels of an extended subhypergraph (and provide them as an additional 
parameter of the function \Decomp), etc. These ideas and several further improvements -- together with the 
pseudo-code of the  resulting improved algorithm -- are 
presented in Appendix~\ref{sect:improvements}.

\section{Implementation and Evaluation}
\label{sect:empirical}

We report now on the empirical results 
obtained for our implementation of the \logk algorithm.
Our experiments are based on the HyperBench benchmark from \cite{10.1145/3440015}, 
which was already used for the evaluation of previous decomposition algorithms, notably 
\newdetk~\cite{10.1145/3440015} 
(an enhanced re-implementation of \detk~\cite{DBLP:journals/jea/GottlobS08})
and \hdsmt~\cite{DBLP:conf/ijcai/SchidlerS21}.

Our goal was to determine the exact hypertree width of as many instances as possible. 
We compare here the performance of three different decomposition methods,
namely 
\newdetk~\cite{10.1145/3440015}, \hdsmt~\cite{DBLP:conf/ijcai/SchidlerS21},
and our implementation of \logk. Note that while the tested implementations include the capability to compute GHDs or FHDs, we only consider the computation of HDs in our experiments here. 
Our new implementation of \logk is based on the open-source code of BalancedGo~\cite{DBLP:conf/ijcai/GottlobOP20}, a parallel algorithm for computing GHDs.
The full raw data of our experiments\footnote{\url{https://zenodo.org/record/6389816}} as well as the  source code of our implementation\footnote{\url{https://github.com/cem-okulmus/log-k-decomp}} of \logk are provided at the URLs below.

\subsection{Benchmark Instances and Setting}

For the evaluation, we use the benchmark library HyperBench~\cite{10.1145/3440015}. It contains 3648 hypergraphs underlying CQs and CSPS from various sources in industry and the literature and is commonly used to evaluate decomposition algorithms. 
The instances are available at \url{http://hyperbench.dbai.tuwien.ac.at} and in the raw data accompanying this manuscript. %

\paragraph{Hardware and Software.}
Our implementation is written in the programming language Go using version 1.14 and we will refer to it as \logk. We will give more details below on how it was configured for the experiments reported in Section~\ref{subsect:evalorsomething}.  The hardware used for the evaluation was a cluster of 12 nodes, using Ubuntu 16.04.1 LTS, with Linux kernel 4.4.0-184-generic, GCC version 5.4.0. Each node has a 12 core Intel Xeon CPU E5-2650 v4, clocked at  2.20 GHz and using 264 GB of RAM. 

\paragraph{Setup of Experiments.} 
To ensure comparability of our experiments with results published in the literature, we employ the following test  setup and restrictions:
a timeout of one hour was used and available RAM was limited to 1 GB. We note that this corresponds to limits also used in previous experiments in this area~\cite{10.1145/3440015,DBLP:conf/ijcai/GottlobOP20}.
For \logk, each run needs two inputs: a hypergraph $H$ and the width parameter $k \geq 1$. 
For these tests, we used width parameters in the range $[1,10]$. 
When running tests for \hdsmt, we used different memory limits. Namely, we allowed \hdsmt to use up to 24 GB of RAM since 
SMT solving is significantly more memory intensive than the other two algorithms. Note that the other two algorithms have very low memory requirements and are not constrained in any way by the 1 GB limit and the respective experiments are therefore still comparable with \hdsmt.
Furthermore, \hdsmt needs no width parameter since it directly tries to find an optimal solution.

We used the HTCondor system~\cite{condor-practice} to facilitate the tests, limits to memory and number of cores accessed by running test instances. 

Throughout this section we will be interested in two key
  metrics. First, the number of \emph{solved} instances, by which we
  mean instances for which an optimal (i.e., minimal width) hypertree
  decomposition was found and proven optimal. Second, the computation
  time that was necessary to compute the optimal width decomposition,
  which we will refer to as the \emph{running time} or simply
  \emph{runtime}. Importantly, this means that average running times
  are taken only over the instances that the respective algorithm is able to solve,
  while timed out instances are not considered in the running time calculation.

\begin{table*}[t]
\setlength{\tabcolsep}{2pt}
\centering

\newcolumntype{?}{!{\vrule width 1pt}}

    \caption{Comparison of prior methods and \logk: number of cases\textbf{} solved and runtimes (sec.) to find  optimal-width HDs.}
    \label{tab:bigeval}
    
\begin{tabular}{c?  c | c| c}
\toprule
  & \multicolumn{3}{c}{Hypertree Decomposition Methods} \\ 
\midrule
\begin{tabular}{c c}
\begin{tabular}{r r c}
    \multicolumn{1}{c}{Origin of} & \multicolumn{1}{c}{Size of} &  \multicolumn{1}{c}{Instances in}  \\
    \multicolumn{1}{c}{Instances}  & \multicolumn{1}{c}{Instances} &  \multicolumn{1}{c}{Group} \\
 \midrule

Application & $75< |E| \leq 100$            & 405 \\ 
            & $50< |E| \leq \phantom{0}75 $ & 514 \\   
            & $10< |E| \leq \phantom{0}50 $ & 369 \\ 
            &     $|E| \leq \phantom{0}10 $  & 915 \\ 
\midrule 

Synthetic   & $\phantom{10 < } |E| > 100 $              &  \phantom{0}66 \\  
            &             $75< |E| \leq 100$            & 422 \\ 
            &             $50< |E| \leq \phantom{0}75 $ & 215 \\   
            &             $10< |E| \leq \phantom{0}50 $ & 647 \\ 
            & $\phantom{100<}  |E| \leq \phantom{0}10 $ &  \phantom{0}95 \\ 
\midrule
\multicolumn{1}{c}{Total}   & \multicolumn{1}{c}{-} & 3648 
\end{tabular}

 \end{tabular} &  
 \begin{tabular}{rrrr}

\multicolumn{4}{c}{\newdetk~\cite{10.1145/3440015}} \\  %

 \#solved & avg & max & stdev \\ 
\midrule
97&    %
21.4  & 3296.0  & 192.8 \\

276&    %
10.6  & 1906.0  & 104.7\\

253&    %
  0.0 & 0.0 & 0.0\\

906&    %
  0.0 & 0.0 & 0.0\\

  \midrule 

18&     %
  0.2 & 7.0 & 1.0\\

87&    %
  77.2 & 3467.0  & 379.3\\
 
38&    %
 18.8 & 1593.0  & 141.9 \\

290&    %
   56.0 & 3240.0  & 336.3  \\

\textbf{95}&     %
  0.0 & 0.0 & 0.0 \\

\midrule

2060  & %
20.6  & 3467.0  & 194.2

 \end{tabular}  & 
  \begin{tabular}{rrrr}

\multicolumn{4}{c}{\hdsmt~\cite{DBLP:conf/ijcai/SchidlerS21}} \\  

 \#solved & avg & max & stdev \\ 
\midrule
65&    %
  809.5 & 3156.6  & 735.2  \\

448&    %
 250.0  & 3281.5  & 409.3  \\

237&    %
 60.1 & 1017.9  & 150.3 \\

876&    %
 56.6 & 1427.1  & 155.0  \\

 \midrule 

13&     %
 734.0  & 2507.1  & 711.7   \\

\textbf{312}&    %
 1045.2 & 3591.1  & 1287.0  \\
 
212&    %
    101.7 & 2560.1  & 246.1  \\

303&    %
  412.2 & 3597.4  & 850.2   \\

78&     %
  28.8  & 218.5 & 41.5  \\

\midrule

2544  & %
  280.2 & 3597.4  & 676.7

 \end{tabular} & 
  \begin{tabular}{rrrr}

\multicolumn{4}{c}{ \logk Hybrid} \\  %

 \#solved & avg & max & stdev \\ %
\midrule
\textbf{261}&    %
86.5  & 3555.8  & 332.4 \\

\textbf{469}&    %
0.5 & 78.5  & 3.6 \\

\textbf{253}&    %
 0.0  & 0.1 & 0.0 \\

\textbf{915}&    %
0.0 & 0.0 & 0.0 \\

\midrule 

\textbf{34}&     %
46.9  & 2528.2  & 209.6 \\

235 &    %
48.9  & 2495.6  & 210.9 \\

\textbf{215}&    %
4.1 & 476.3 & 32.7 \\

\textbf{625}&    %
 18.8 & 3526.3  & 174.7  \\

\textbf{95}&     %
0.0 & 0.0 & 0.0 \\

\midrule

\textbf{3102}  & %
30.5  & 3555.8  & 197.8

 \end{tabular}

 \\
\bottomrule
\end{tabular}

 \end{table*}
\vspace{-2mm}

\subsection{Empirical Evaluation}
\label{subsect:evalorsomething}

We report here on the main results of our experiments. A number of additional experiments can be found in Appendix~\ref{sect:moreEval}, 
providing a variety of further details and insights. 
Our implementation of \logk also employs the following {\em hybridisation strategy\/}: 
as will be seen below, \newdetk 
performs very well on small hypergraphs but has difficulties with even slightly larger instances. In contrast, a particular strength of our new 
\logk algorithm is to quickly split a big hypergraph into significantly smaller extended subhypergraphs. To combine the best of both worlds, 
we use \logk to split the original HD computation problem until the subproblems become small, at which point
we apply our own implementation of \detk (extended to handle extended subhypergraphs correctly) to the small subproblems. 
For details and an 
 experimental evaluation of different parametrisations for our hybridisation strategy, %
see Appendix~\ref{sec:hybrid}.

We compare the aforementioned hybrid version of the \logk algorithm %
with the two state-of-the-art implementations for finding HDs: \newdetk~\cite{10.1145/3440015} and 
\hdsmt~\cite{DBLP:conf/ijcai/SchidlerS21}. 

Our results are summarised in Table~\ref{tab:bigeval}, distinguishing the  hypergraphs in the HyperBench benchmark by size and origin \footnote{HyperBench instances are often categorised more fine-grained in terms of their origin (cf.,~\cite{10.1145/3440015}). For our experiments we have found the direct effect of hypergraph size to be more informative and therefore report our results in this way instead.}.
We distinguish between two main categories, hypergraphs that are derived from applications and hypergraphs that were synthetically generated. In each group we report our results split by the number of edges $|E|$ in the instance. Note that the group $|E|>100$ of instances with more than 100 edges is empty for the Application case and thus omitted from the table.
\emph{Instances in Group} reports the number of instances in each such group. 
For each algorithm and each group of instances, we list the number of solved instances (\emph{\#solved}) and statistics over the running times (\emph{avg}, \emph{max}, \emph{stdev}). Times are all in seconds and rounded to a single digit after the comma. Results over all groups are given in the last row titled ``Total''.

As mentioned above, some care is required when comparing times between
algorithms. While \newdetk has low average time overall, this is partly
due to solving fewer instances. The data therefore demonstrates that,
in general, \newdetk either solves an instance quickly or fails to find
an optimal width decomposition before timing out.  Overall, we see that despite
solving significantly more instances than its competitors, running times for \logk overall are
comparable with \newdetk and noticeably lower than for \hdsmt.

It may be of further interest how these numbers compare to the performance of state of the art algorithms for finding generalised hypertree decompositions. The results reported for BalancedGo~\cite{DBLP:conf/ijcai/GottlobOP20} (on a comparable system) show that the best method there solves only 1730 instances optimally without timeout. In contrast \logk manages to solve 2491 of the instances tested there optimally\footnote{The evaluation in \cite{DBLP:conf/ijcai/GottlobOP20} considers only a subset of HyperBench with 3071 instances}. Furthermore, in none of the cases where BalancedGo finds the optimal $\ghw$ is it lower than the optimal $\hw$. 
In other words, in practice, the additional complexity of 
GHDs compared with HDs is not compensated by achieving lower width
(even if, in theory, no better upper bound on the $\hw$ 
than $\hw \leq 3 \cdot \ghw + 1$ is known
\cite{DBLP:journals/ejc/AdlerGG07}). 

In our experiments,
we also observe that for low widths  -- i.e., cases where using HDs is most promising in practice -- \logk is very close 
to solving all instances. In particular, of the 3224 instances with width at most 6, \logk solves 2930 (92\%) instances. 
In contrast,  \newdetk and \hdsmt time out on 1206 and  766, respectively, of those instances. 
This  suggests that \logk can be 
a solid foundation for the integration of HDs in practice going forward. If we look at instances of $\hw \leq 5$ the situation improves even further, with \logk solving $2450$ out of $2482$ ($98.7\%)$ instances;  compared to $80\%$ and $86\%$ solved by \newdetk and \hdsmt, respectively.

The experiments reported in Table~\ref{tab:bigeval}
were performed over the full set of HyperBench instances. However, for the additional experiments
reported in this section, it is more meaningful to restrict our experiments to exclude hypergraphs that are, roughly speaking, too small or have high width. Small instances benefit only marginally from algorithmic improvements or parallelism, while very high width is of less algorithmic interest as it exponentially effects algorithms that make use of decompositions. Hence, we propose to exclude such instances to make more relevant observations.
We therefore focus on instances with more than 50 edges and vertices that are known to have hypertree width at most 6.
  There are 465 instances in HyperBench which satisfy
  these conditions; we will refer to them as \hblarge.

We performed a second set of experiments over the instances in \hblarge to verify our claims that \logk is well-suited for parallelisation. For $1 \leq n\leq 5$, we observe the time taken to find and verify the optimal width of an instance using $n$ CPU cores. We report on the times to find these optimum widths averaged over all instances in \hblarge in Figure~\ref{speedup}. To avoid a decreasing number of timeouts from skewing the data we report the average only over instances that do not timeout for any $n$ for a given algorithm. For reference, we also report the (single core) performance of \newdetk for the same setting. 

We observe approximately linear speedups up to 4 cores, from about 189 seconds on 1 core to 50 seconds for 4 cores for \logk.
This behaviour is expected since
 our parallelisation strategy relies on dividing up the search space for bounded separators uniformly over the the available cores. Since this requires no communication between threads or other overhead that depends on the degree of parallelisation, the key task of searching for balanced separators  scales linearly in the number of cores.
 In instances where the search for separators dominates the running time, such as negative instances where the full search space is explored, analysis of our algorithm therefore predicts effectively linear scaling of performance.
 In the data from Figure~\ref{speedup}, we observe diminishing returns in the rate of improvement of average running time starting from 5 cores. However, preliminary experiments on additional different systems do not confirm this behaviour and there \logk exhibits linear scaling up to much higher core counts. Further in-depth experimentation is therefore required to obtain a clearer picture for the scaling behaviour for a high number of cores.

Very similar scaling can be observed for our Hybrid version. Note that the reported times for the Hybrid algorithm are slightly higher only due to solving more (harder) instances. %

\begin{figure}[t]
  \definecolor{fancyblue}{HTML}{298CE3}
  \definecolor{fancyblueNeg}{HTML}{0F4C81}
  \definecolor{vermillion}{HTML}{E34234}
  \definecolor{vermillionNeg}{HTML}{69150E}

\begin{minipage}[]{5cm}
 \centering
  \begin{tikzpicture} [transform shape,scale=0.75]
\begin{axis}[ 
  ylabel={Average running times (sec)},
    xtick = {1,...,6},
    width = 7 cm,
    height= 8 cm,
  ]

\addplot[color=fancyblue, mark=square]  coordinates {

(1 , 189.76)
(2 , 101.56)
(3 , 69.23)
(4 , 51.87)
(5 , 43.03)
(6 , 41.17)

};
\addlegendentry{\logkshort \phantom{(Hybrid) \phantom{neg.}}}; %

\draw [gray] (0,124.47) --(8,124.47)    
node[pos = 0.55,above] {\newdetk};

\addplot[color=vermillion, mark=square]  coordinates {

(1 , 172.12)
(2 , 89.76)
(3 , 62.87)
(4 , 50.37)
(5 , 43.1)
(6 , 40.49)

};
\addlegendentry{\logkshort (Hybrid) \phantom{neg.} };

  \path (0.5,0) coordinate (O) (4.4,0) coordinate (R);  
\end{axis}

\node at (2.7, -0.7)  {\# cores};

\end{tikzpicture} 
\end{minipage} \begin{minipage}[]{3cm}  
 \centering
 \small
\setlength{\tabcolsep}{2pt}
\begin{tabular}{l c}
\toprule 
Method & Timeouts  \\
\midrule 
\logkshort (Hybrid) &  \textbf{143} \\ 
\logkshort  & 666 \\ 
\newdetk & 611 \\ 
\bottomrule
\end{tabular}

\end{minipage}

\caption{Study of \logk scaling behaviour w.r.t. the number of processing cores used.
}
\label{speedup}
\end{figure}
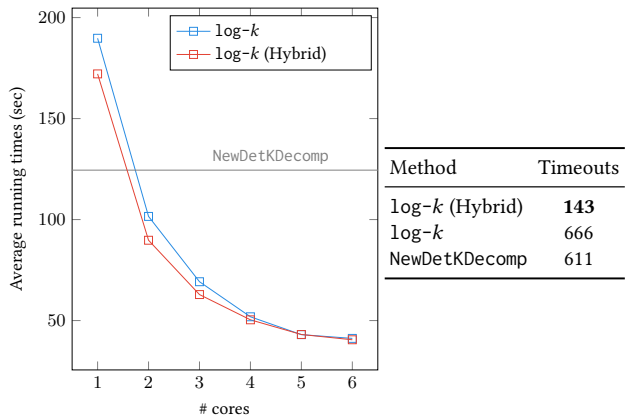

\section{Conclusion}
\label{sect:Conclusion}
In this paper we introduced a novel algorithm \logk for 
computing hypertree decompositions. Based on new theoretical
insights and results on HDs, we were able to propose an algorithm that constructs decompositions in
arbitrary order (rather than, e.g., in a strict top-down manner) while achieving a \emph{balanced} separation into 
subproblems. In this way, we have obtained a 
logarithmic bound on the 
recursion depth of our algorithm, making it particularly 
well suited 
for parallelisation.
We evaluated an implementation of \logk through experimental
comparison with the state of the art. On the standard benchmark for
hypertree decomposition, we are able to achieve clear improvements both in the number of solved instances 
and in the time required to solve them. 

In combination, our theoretical results and experiments demonstrate
that \logk achieves our goal of effective parallel HD computation. We
believe that the performance improvements, especially on large
hypergraphs lay a strong foundation
for more widespread adoption of hypertree decompositions in practice, e.g., for complex query
execution in high-performance database applications.

With HD computation for large and complex hypergraphs becoming practically
feasible, one of the key challenges that block the use of HDs is quickly
becoming less problematic.  We therefore consider full integration of
hypertree decompositions into existing database systems and constraint solvers to be a
natural next step in this line of research.

Experiments suggest that there is
significant potential in the study of metrics for hybrid
approaches. In particular, how can we decide effectively when to switch from the balanced separation of \logk to the greedy heuristic guided method underlying \detk. This motivates a more in-depth study of hybridisation metrics in the future.

\section*{Acknowledgements}
This work was supported by the Austrian Science Fund (FWF) project P30930-N35. Georg Gottlob is a Royal
Society Research Professor and acknowledges support by the Royal Society for the present work in the
context of the project "RAISON DATA" (Project reference: RP\textbackslash R1\textbackslash 201074).
Matthias Lanzinger acknowledges support by the Royal Society  project "RAISON DATA" (Project reference: RP\textbackslash R1\textbackslash 201074).

\bibliographystyle{abbrv}
\bibliography{soda_parallelHD}

\appendix

\section{Correctness Proof of Algorithm \ref{alg:parHD}}
\label{sect:correctness}

We prove the soundness and completeness of the algorithm \logk given in
 Algorithm~\ref{alg:parHD} separately. 
The polyno\-mial-time upper bound on the construction of an HD in case of a successful run 
of the algorithm (i.e., if it returns ``true'') will be part of the soundness proof.

It is convenient to first prove the following claim: 

\medskip
\noindent
{\sc Claim A.}
{\em In every call of function \Decomp in Algorithm~\ref{alg:parHD} with parameters
$(H', \Con)$ it is guaranteed that 
$\Con \subseteq V(H')$ holds with 
$V(H') = \big(\bigcup H'.E) \cup \big(\bigcup H'.\Sp)$.
}

\begin{proof}[Proof of Claim A]
The proof is by induction on the call depth of the recursive function \Decomp.

\smallskip

\noindent{\em induction begin.} The top-level calls of function \Decomp on line 7 are with
 parameters $(y,\Con_y)$, where $\Con_y$ is defined on line 6 as $\Con_y = V
 (y) \cap \bigcup \lambda(r)$. Hence, $\Con_y \subseteq V(y)$ clearly holds. 

\smallskip

\noindent{\em induction step.} Suppose that Claim A holds for every call of function \Decomp down to
 some call level $n$ and suppose that function \Decomp is called recursively during execution
 of \Decomp at call level $n$. Suppose that this execution of \Decomp is with parameters $
 (H',\Con)$. The only places where function \Decomp is called recursively are lines 33 and 37. More
 specifically, \Decomp is called with parameters $(x,\Con_x)$ on line 33 and with parameters $
 (comp_{up}, \Con)$ on line 37. We have to show that both $\Con_x \subseteq V(x)$ (on line 33) and
 $\Con \subseteq V(comp_{up})$ (on line 37) hold. On line 33, the condition is trivially fulfilled,
 since $\Con_x$ is defined on line 32 as $\Con_x = V(x) \cap \chi(c)$. 

It remains to  consider the call of function \Decomp on line~37. 
Suppose that $comps_c$ on line 28 is of the form 
$comps_c = \{x_1, \dots, x_\ell\}$. 
By the definition of components in Definition~\ref{def:connectedness-components}, 
$H'$ (that is, $H'.E \cup H'.\Sp$) can be partitioned into the following 
disjoint subsets:
\begin{itemize}[topsep=5pt]
\item $x_1.E \cup x_1.\Sp, \dots, x_\ell.E \cup x_\ell.\Sp$
\item $y = \{f \in H'.E  \cup H'.\Sp \mid f \subseteq \chi(c)\}$. 
\item $z = (H'.E \setminus comp_{\low}.E) \cup (H'.\Sp \setminus comp_{\low}.\Sp)$
\end{itemize}
We thus have $V(H') = \bigcup_{i=1}^\ell V(x_i) \cup V(y) \cup V(z)$ with $V(y) \subseteq \chi
(c)$. By construction (line 28), all components $x_i$ are contained in $comp_{\low}$. Hence, we
actually have $V(H') = V(comp_{\low}) \cup \chi(c)  \cup V(z)$. The recursive call of
function \Decomp on line 37 is with the edges and special edges in $z$ plus $\chi(c)$ as an
additional special edge. Hence, we have $V(comp_{up}) = V(z) \cup \chi(c)$ when \Decomp is called
on line 37 with parameters $(comp_{up},\Con)$. It is therefore sufficient to show that
$\Con \subseteq V(z) \cup \chi(c)$ holds. 

By the induction hypothesis, we may assume that $\Con \subseteq V(H')$ holds. The check on line 22
ensures that $\Con \subseteq \bigcup \lambda(p)$ holds for some edge set $\lambda(p)$. Moreover,
the check on line 26 ensures that $\big(\bigcup \lambda(p)) \cap V(comp_{\low}) \subseteq \chi(c)$. In
total, we thus have $\Con  \cap V(comp_{\low}) \subseteq \chi(c)$. Together with $V(H') = V(comp_
{\low}) \cup \chi(c)  \cup V(z)$ and $\Con \subseteq V(H')$, we may thus conclude $\Con \subseteq V
(z) \cup \chi(c)$ and, therefore, $\Con \subseteq V(comp_{up})$ (on line 37). Hence, also
the call of function \Decomp on line 37 satisfies Claim A.
\end{proof}

\begin{proof}[Soundness Proof]
Suppose that algorithm \logk returns ``true''. 
That is, for some value of $\lambda(r)$,  
each call of function \Decomp on line 7 returns ``true''. 
We have to show that there exists an HD of width $\leq k$ of $H$. 
To construct such an HD, we take
$\lambda(r)$ as the $\lambda$-label of the root $r$ of this HD. 
By the special condition of HDs, we must take 
$\chi(r) = \bigcup \lambda(r)$. Hence, the $[\chi(r)]$-components and
$[\lambda(r)]$-components of $H$ coincide 
and $comps_r$ computed on line 4 contains all $[\chi(r)]$-components of $H$. 

Now suppose that function \Decomp is sound (we will prove the correctness of this assumption below),
i.e., for an arbitrary extended subhypergraph $(E',\Sp,\Con)$ of $H$, if function \Decomp returns
``true'' on input $((E',\Sp),\Con)$, then there exists an HD of width $\leq k$ of $
(E',\Sp,\Con)$. Hence, if the calls of function \Decomp on line 7 all yield ``true'', then, by
assuming the soundness of \Decomp, we may conclude that an HD of width $\leq k$ exists for each
extended subhypergraph of $H$ of the form $(y.E, y.\Sp, \Con_y)$. Let each of these HDs be denoted
by ${   \mathcal{D} }[y]$ with tree structure $T[y]$ and let $r[y]$ denote  the root of $T[y]$. Then we can
construct an HD of $H$ by taking $r$ with $\lambda(r)$ from line 3 and $\chi(r)= \bigcup \lambda
(r)$ as root node and appending the HD-fragments ${ \mathcal{D}  }[y]$ to $r$, such that the root nodes $r
[y]$ of the trees $T[y]$ become child nodes of $r$. It is easy to verify that the resulting
decomposition is an HD of $H$, and this HD can be constructed in polynomial time from the
HD-fragments ${ \mathcal{D}  }[y]$. It remains to show that function \Decomp is sound.

\smallskip

\noindent{\em Soundness of function \Decomp.} 
For an arbitrary extended subhypergraph $(E',\Sp,\Con)$ of $H$, 
let function \Decomp return ``true'' on input 
$((E',\Sp),\Con)$; we have to show that then there exists an HD of width $\leq k$ of $(E',\Sp,\Con)$.
Moreover, we have to show that 
by materialising the decompositions implicitly constructed in the recursive calls of 
function \Decomp, 
an HD  of width $\leq k$ of $(E',\Sp,\Con)$  can be constructed in polynomial time
whenever  \Decomp returns ``true''.
The proof is by induction on $|E'| + |\Sp|$. 

\smallskip

\noindent{\em induction begin.} Suppose that $|E'| + |\Sp| = 1$ and that function \Decomp returns
 ``true''. Hence,   we either have $|E'| = 1$ and $|\Sp| = 0$ or we have $|E'| = 0$
 and $|\Sp| = 1$. In either case, an HD of this extended subhypergraph can be obtained with a
 single node $u$ by setting $\lambda(u) = \{f\}$ and $\chi(u) = f$, where $f$ is the only
 (special) edge in $E' \cup \Sp$. This decomposition clearly satisfies all conditions of an HD
 according to Definition~\ref{def:extendedHD}, the only non-trivial part being Condition (6): we
 have to verify $\Con \subseteq \chi(u)$. By Claim~ A above, we know that in every
 call of function \Decomp, $\Con$ is a subset of the vertices in $E' \cup \Sp$. Now in case $|E'| +
 |\Sp| = 1$ holds, we have $E' \cup \Sp = \{f\}$ for a single (special) edge $f$ and, therefore,
 $\chi(u) = f = \big(\bigcup E'\big) \cup \big(\bigcup \Sp\big)$. Hence, we indeed have
 $\Con \subseteq \chi(u)$. 

\medskip

\noindent
{\em induction step.}
Now suppose that $|E'| + |\Sp| > 1$ and that function \Decomp returns ``true''. 
This means that one of the return-statements in lines 13, 15, or 39 is executed. Actually, 
line 15 can be excluded for $|E'| + |\Sp| > 1$. Now consider the remaining two lines 13 and 39. 
If the return-statement on line  13 is executed, then 
we have $|E'| \leq k$ and $|\Sp| = 0$. In this case, analogously to the induction begin, the desired HD 
consists of a single node $u$ with $\lambda(u) = E'$ and $\chi(u) = \bigcup E'$. Again, 
all conditions of an HD according to Definition~\ref{def:extendedHD} are easy to verify; in particular,
the proof argument for Condition (6) is the same as above. 

It remains to consider the case that ``true'' is returned on line 39. This means that, for a
particular value of $\lambda(p)$ (chosen on line 16) and of $\lambda(c)$ (chosen on line 24), all
recursive calls of function \Decomp (on lines 33 and 37) return ``true''. By the induction
hypothesis, we may assume that for each of the extended subhypergraphs processed by these
recursive calls of \Decomp, an HD of width $\leq k$ exists. Note that we are making use of Claim~A
here in that we may assume that all recursive calls of \Decomp are with properly defined extended
subhypergraphs (in particular, the vertex set supplied as second parameter is covered by the edges
and special edges in the first parameter of each such call). 
Now look at these recursive calls: we are studying a call of function \Decomp with parameters
$H'$ and $\Con$, where $H'$ consists of a set $H'.E$ of edges and a set $H'.\Sp$ of special edges.
That is, function \Decomp  is processing the extended subhypergraph $(H'.E, H'.\Sp, \Con)$. The
current call of function \Decomp apparently has chosen labels $\lambda(p)$ and $\lambda(c)$ for
nodes $p$ and $c$, such that all checks on lines 18, 22, 26, and 29 are successful
in the sense that program execution continues with these values of $\lambda(p)$ and $\lambda(c)$.
In particular,
there exists a $[\lambda(p)]$-component $comp_{\low}$ of $(H'.E, H'.\Sp, \Con)$, satisfying the
conditions $V(comp_{\low}) \cap \Con \subseteq \bigcup \lambda(p)$   (line 22) and 
$V(comp_{\low}) \cap \bigcup\lambda(p)  \subseteq  \chi(c) $~(line~26). 

Let $\{x_1, \dots, x_\ell\}$ denote the set of $[\chi(c)]$-components of 
$H'$  inside $comp_{\low}$. Then $H'.E \cup H'.\Sp$ (i.e., the set of edges and
special edges in $H'$) can be partitioned into the following disjoint subsets:
\begin{itemize}[topsep=5pt]
\item $x_1.E \cup x_1.\Sp, \dots, x_\ell.E \cup x_\ell.\Sp$
\item $(H'.E \setminus comp_{\low}.E) \cup (H'.\Sp \setminus comp_{\low}.\Sp)$
\item $\{f \in H'.E  \cup H'.\Sp \mid f \subseteq \chi(c)\}$. 
\end{itemize}
From the first two kinds of sets of  edges and special edges, 
the following extended subhypergraphs are constructed, for which function \Decomp is then called
recursively on lines 33 and 37: 
\begin{itemize}[topsep=5pt]
\item for  each $x_i$ consisting of a set of edges $x_i.E$ and special edges $x_i.\Sp$, 
define $H_i = (x_i.E, x_i.\Sp, \Con_i)$ with 
$\Con_i = V(x_i) \cap \chi(c)$; 
\item for $(H'.E \setminus comp_{\low}.E) \cup (H'.\Sp \setminus comp_{\low}.\Sp)$
define $H^\uparrow = (E^\uparrow, \Sp^\uparrow, \Con^\uparrow)$ with 
$E^\uparrow = H'.E \setminus comp_{\low}.E)$ and 
$\Sp^\uparrow = (H'.\Sp \setminus comp_{\low}.\Sp) \cup \{\chi(c)\}$
and $\Con^\uparrow = \Con$.
\end{itemize}
By assumption, the recursive calls of \Decomp for each of these extended subhypergraphs return the
value ``true''. Thus, by the induction hypothesis, for each of these extended subhypergraphs, there
exists an HD of width $\leq k$. From these HDs, we construct an HD of $(H'.E, H'.\Sp, \Con)$ as
follows: 
\begin{itemize}[topsep=5pt]
\item First take the HD of $H^\uparrow$. We shall refer to this HD as ${ \mathcal{D}  }^\uparrow$. Let $r$
 denote the root node of ${ \mathcal{D}  }^\uparrow$. By $\Con^\uparrow = \Con$, we have
 $\Con \subseteq \chi(r)$. 
\item Recall that $\chi(c)$ was added as a special edge to the extended subhypergaph $H^\uparrow$.
 Hence, by Definition~\ref{def:extendedHD}, the HD ${ \mathcal{D}  }^\uparrow$ has a leaf node $u$ with
 $\lambda(u) = \{\chi(c)\}$ and $\chi(u) = \chi(c)$. Now we replace node $u$ in ${ \mathcal{D}   }^\uparrow$
 by node $c$ with $\lambda(c)$ and $\chi(c)$ according to the current execution of
 function \Decomp. Moreover, for every $f \in H'.\Sp$ with $f \subseteq \chi(c)$, we append a fresh
 child node $c_f$ to $c$ with $\lambda(c_f) = \{f\}$ and  $\chi(c_f) = f$. It is easy to verify
 that the resulting decomposition (let us call it ${ \mathcal{D} }'$) is an HD of the extended
 subhypergraph that contains all edges and special edges of $H'$ except for the ones in 
 any of the $x_i$'s. In particular, node $r$ with $\Con \subseteq \chi(r)$  is still the 
 root of~HD~${\mathcal{D}}'$. 
\item Now we take the HDs ${ \mathcal{D} }_i$ of the extended subhypergraphs $(x_i.E, x_i.\Sp, \Con_i)$ and
 append them as subtrees below $c$ in ${ \mathcal{D} }'$, i.e.: the root nodes of the HDs ${ \mathcal{D} }_i$
 become child nodes of $c$. Let us refer to the resulting decomposition as ${ \mathcal{D} }$. It remains
 to show that ${ \mathcal{D} }$ indeed is an HD of width $\leq k$ of the extended subhypergraph $
 (H'.E, H'.\Sp, \Con)$ of $H$. The width is clear, since all HD-fragments of ${ \mathcal{D} }$ and also
 $\lambda(c)$ have width $\leq k$. It is also easy to verify that every edge in $H'.E$  is covered
 by some node in ${ \mathcal{D} }$ and every special edge in $H'.\Sp$ is covered by some leaf node in $
 { \mathcal{D} }$. Moreover, also the connectedness condition holds inside each HD-fragment (by the
 induction hypothesis) and between the various HD-fragments. The latter condition is ensured by
 the definition of components in 
 Definition~\ref{def:connectedness-components} and by the fact that any two extended 
 subhypergraphs processed by the various recursive calls of function \Decomp can only share vertices
 from $\chi (c)$.
\end{itemize}
Finally,  note that the above construction of HD ${ \mathcal{D} }$ from the HD-fragments constructed in the
recursive calls of \Decomp is clearly feasible in polynomial time. 
\end{proof}

Before we prove the completeness of 
algorithm \logk, we introduce a special kind of 
extended subhypergraphs: let $H$ be a hypergraph and let 
$\mathcal{D} = \langle T,\chi,\lambda\rangle$ be an
HD of $H$ with root $r$. We call $H' = (E',\Sp,\Con)$ a {\em  $\mathcal{D}$-induced extended 
subhypergraph\/} of $H$, if there exists a subtree $T'$ of $T$ with the following properties:

\begin{itemize}[topsep=5pt]
\item $E' = \cov(T')$;
\item let $B$ denote those nodes in $T$ which are outside $T'$ but whose parent node is in 
$T'$; then $\Sp = \{ \chi(u) \mid u \in B\}$.
\item the root $r'$ of $T'$ is different from the root of $T$; hence, $r'$ has a parent node $p$ in $T$;
\item $\Con = V(H') \cap \bigcup \lambda(p)$.
\end{itemize}
An HD $\mathcal{D}' = \langle S', \chi', \lambda' \rangle$ of $H'$  is then obtained as follows: 
\begin{itemize}[topsep=5pt]
\item the tree $S'$ of $\mathcal{D}'$ is the subtree of $T$ induced by the nodes of $T'$ plus 
the nodes in $B$;
\item for all nodes u in $T'$, we set $\chi'(u) = \chi(u)$ and $\lambda'(u) = \lambda(u)$;
\item for all nodes u in $B$, we set $\chi'(u) = \chi(u)$ and $\lambda'(u) = \{\chi(u)\}$.
\end{itemize}

We shall refer to $\mathcal{D}'$ as the induced HD of $H'$.
It is easy to verify that $\mathcal{D}'$ is in normal form, whenever
$\mathcal{D}$ is in normal form.
In the completeness proof below, 
we shall refer to a  $\mathcal{D}$-induced extended subhypergraph of $H$
simply as an ``induced subhypergraph'' of $H$. 
No confusion can arise from this, since 
we will  always consider the same 
HD $\mathcal{D}$ of $H$ throughout the proof.

\begin{proof}[Completeness Proof] 
Suppose that hypergraph $H$ has an HD of width
 $\leq k$. We have to show that then algorithm \logk returns ``true''. 
 By Theorem~\ref{theo:normalform}, $H$
 also has an HD $\mathcal{D} = \langle T, \chi$, $\lambda \rangle$ of width $\leq k$ in normal
 form. Note that, in order to apply Theorem~\ref{theo:normalform}, we are considering $H$ as an
 extended subhypergraph $(E(H), \emptyset, \emptyset)$ of itself. Let $r$ denote the root of $T$.
 If algorithm \logk has not already returned ``true'' before, it will eventually try
 $\lambda(r)$ in the foreach-statement on line 3. Let $C_1, \dots, C_\ell$ with 
 $C_i \subseteq E (H)$ for each $i$ denote the $[\lambda(r)]$-components 
 (and hence also the $[\chi (r)]$-components) of $H$. 
 By the normal form of $\mathcal{D}$, we know that $r$ has $\ell$ child
 nodes $u_1, \dots, u_\ell$ with $\cov(T_{u_i}) = C_i$. 
 Now let $\Con_i = V(C_i) \cap \bigcup \lambda(r)$ for $i \in \{1, \dots, \ell\}$. 
Hence, $(C_i, \emptyset, \Con_i)$ is an extended subhypergraph of $H$. Moreover, 
by the connectedness condition, we have $\Con_i \subseteq \chi(u_i)$.
Hence, HD $\mathcal{D}$ restricted to the
subtree $T_{u_i}$ is in fact an HD of width $\leq k$ of the extended subhypergraph 
$(C_i, \emptyset, \Con_i)$. 

Now suppose that function \Decomp is complete 
on induced subhypergraphs of $H$
(we will prove the correctness of this assumption below). By this we mean that 
if $(E',\Sp,\Con)$ is an induced subhypergraph of $H$, 
then function \Decomp returns ``true'' on input 
$(E',\Sp),\Con)$.
Hence, the calls of function \Decomp on line 7 all yield ``true''.
Therefore, program execution exits the foreach-loop and 
executes the return-statement on line 9. 
That is, algorithm \logk returns ``true'' as desired. 
It remains to show that function \Decomp is complete on induced subhypergraphs.

\smallskip

\noindent{\em Completeness of function \Decomp.} 
Consider an arbitrary induced subhypergraph $(E',\Sp,\Con)$ of $H$. 
We have to show that then function \Decomp returns ``true'' on input 
$((E',\Sp),\Con)$. We proceed by induction 
on $|E'| + |\Sp|$.

\smallskip

\noindent
{\em induction begin.}
Suppose that $|E'| + |\Sp| = 1$. That is, 
we either have 
$|E'| = 1$ and $|\Sp| = 0$ or we have 
$|E'| = 0$ and $|\Sp| = 1$. In the first case, 
``true' is returned via the statement on line 13; in the second case, 
``true' is returned via the statement on line 15.

\medskip

\noindent
{\em induction step.}
Now suppose that $|E'| + |\Sp| > 1$.
If $|E'| \leq k$ and $|\Sp| = 0$, then the return-statement on line 13 is executed and the function 
returns ``true''.  It remains to consider the case 
that $|E'| > k$ or $|\Sp| > 1$ holds.
By Lemma~\ref{lem:existence-balsep}, 
the HD $\mathcal{D}'$ induced by $H'$ has a balanced separator. 
Let us refer to this balanced separator as the node $c$ in $\mathcal{D}'$.
By the balancedness, it can be easily verified that $c$ must satisfy $\lambda(c)  \subseteq E(H)$ 
(that is, $c$ is not a leaf node with $\lambda(c) = \{f\}$ for some special edge $f$). 
We distinguish two cases:

\smallskip

\noindent
{\em Case 1.} Suppose that $c$ is the root node of $\mathcal{D}'$. 
Recall that in our definition of induced subhypergraphs, the root of the corresponding tree $T'$ is
different from the root $r$ of $\mathcal{D}$. Hence, $c$ has a parent node in $\mathcal{D}$. 
Let us refer to this parent node as $p$. 
If function \Decomp has not already returned ``true'' before, it will eventually
try $\lambda(p)$ in the foreach-statement on line 16. Due to the normal form of $\mathcal{D}$,
all of $H'$ is a single $[\lambda(p)]$-component. Hence, the if-condition on line 18 is satisfied
and $comp_{\low}$ is assigned all of $H'$ on line 19.
The connectedness check on line 22 succeeds, since $p$ is the parent 
of the root of $\mathcal{D}'$ and $\Con = V(H') \cap \bigcup \lambda(p)$ holds by
the last condition of the definition of induced subhypergraphs.
Hence, the foreach-loop on lines 24 -- 39 is eventually entered.  
If function \Decomp does not return ``true'' before, it will eventually
try $\lambda(c)$ in the foreach-statement on line 24.
Then $\chi(c)$ assigned on line 25 is the correct $\chi$-label of $c$ according to the normal form.
The connectedness check on line 26 succeeds since $\mathcal{D}$ satisfies the connectedness
condition.
By assumption, $c$ is a balanced separator; hence also the check on line 29 succeeds. 
Thus, the foreach-loop on lines 31 --34 is executed. It is easy to verify that the parameters 
supplied to \Decomp in the recursive calls on line 33 correspond to induced subhypergraphs. 
Therefore, all these calls of \Decomp return ``true'' by the induction hypothesis. Hence, also the 
statements on lines 35 -- 37 are executed. 
In this case, since $comp_{\low}$ comprises all edges and special edges of $H'$, 
\Decomp is called on line 37 with $comp_{up}.E = \emptyset$ and 
$comp_{up}.\Sp = \{\chi(c)\}$. Hence, as was shown in the induction begin, 
this call of \Decomp  returns ``true''.
Therefore, the return-statement on line 39 is executed and the overall result ``true'' is 
returned by function \Decomp.

\smallskip

\noindent
{\em Case 2.} Suppose that $c$ is not the root node of $\mathcal{D}'$. 
Then $c$ has a parent node inside  $\mathcal{D}'$. 
Let us refer to this parent node as $p$. 
If function \Decomp has not already returned ``true'' before, it will eventually
try $\lambda(p)$ in the foreach-statement on line 16. 
By Corollary~\ref{cor:chi-vs-lambda}, one of the $[\lambda(p)]$-components of $H'$ 
is the $[\chi(p)]$-component consisting of the edges and special edges which are covered 
in $\mathcal{D}$ by the subtree rooted at child node $c$ of $p$.
The check on line 18 is successful because the child $c$ of $p$ is a balanced separator. Hence 
$c$ and the subtrees below $c$ cover more than half of the edges and special edges. 
The check on line 22 succeeds because of the connectedness condition in $\mathcal{D}$.
Hence, the foreach-loop on lines 24 -- 39 is eventually entered.  
If function \Decomp does not return ``true'' before, it will eventually
try $\lambda(c)$ in the foreach-statement on line 24.
Then $\chi(c)$ assigned on line 25 is the correct $\chi$-label of $c$ according to the normal form.
The connectedness check on line 26 succeeds since $\mathcal{D}$ satisfies the connectedness
condition. By assumption, $c$ is a balanced separator; hence also the check on line 29 succeeds. 
Thus, the foreach-loop on lines 31 --34 is executed. It is easy to verify that the parameters 
supplied to \Decomp in the recursive calls on line 33 correspond to induced subhypergraphs. 
Hence, all these calls of \Decomp return ``true'' by the induction hypothesis. Hence, also the 
statements on lines 35 -- 37 are executed. 
Again, it is easy to verify that also 
the parameters 
supplied to \Decomp in the recursive call on line 37 correspond to an induced subhypergraph. 
Hence, by the induction hypothesis, also 
this call of \Decomp  returns ``true''.
Therefore, the return-statement on line 39 is executed and the overall result ``true'' is 
returned by function \Decomp. 
\end{proof}

\section{An Illustrative Example}
\label{sect:beispiel}
To illustrate the notions introduced in Section~\ref{sect:basics} 
and the  basic algorithm \logk shown in Algorithm \ref{alg:parHD},
we consider the hypergraph $H = (V,E)$ with 
$V = \{x_1, \dots, x_{10}\}$ and 
\begin{align*}
		E = \{&R_1(x_1,x_2), \\ 
&R_2(x_2,x_3), \\ 
&R_3(x_3,x_4), \\ 
&R_4(x_4,x_5), \\ 
&R_5(x_5,x_6), \\ 
&R_6(x_6,x_7), \\ 
&R_7(x_7,x_8), \\ 
&R_8(x_8,x_9), \\ 
&R_9(x_9,x_{10}), \\ 
&R_{10}(x_{10},x_1)\}
\end{align*}
In other words, $H$ is a essentially a cycle of size 10.
A hypertree decomposition $\mathcal{D}$ of $H$ is shown in 
Figure \ref{fig:beispiel}.

\begin{figure*}

\begin{subfigure}[b]{0.3\textwidth}
\centering
\begin{tikzpicture}
[sibling distance=10em,level distance=4em,
every node/.style={shape=rectangle,draw,align=center}]
\node{
	\begin{tabular}{rc}
	$u_1$:    & \multirow{2}{*}{
	$\begin{aligned}
	\lambda &= \{R_1, R_2 \} \\
	\chi &=  \{ x_1,x_2,x_3 \}
	\end{aligned}$
	}\\
	\\
	\end{tabular}
}
child{
	node{
		\begin{tabular}{r c}
		$u_2$:    & \multirow{2}{*}{
		$\begin{aligned}
		\lambda &= \{R_1, R_3 \} \\
		\chi &=  \{ x_1,x_3,x_4 \}
		\end{aligned}$
		}\\
		 \\
		\end{tabular}
	}
	child{
		node{
			\begin{tabular}{r c}
			$u_3$:    & \multirow{2}{*}{
			$\begin{aligned}
			\lambda &= \{R_1, R_4 \} \\
			\chi &=  \{ x_1, x_4,x_5 \}
			\end{aligned}$
				}\\
			\\
			\end{tabular}
		}
		child{
			node{
				\begin{tabular}{r c}
				$u_4$:    & \multirow{2}{*}{
				$\begin{aligned}
				\lambda &= \{R_1, R_5 \}  \\
				\chi &=  \{ x_1, x_5,x_6 \}
				\end{aligned}$
				}\\
				\\
				\end{tabular}
			}
			child{
				node{
					\begin{tabular}{r c}
					$u_5$:    & \multirow{2}{*}{
					$\begin{aligned}
					\lambda &= \{R_1, R_6 \}  \\
					\chi &=  \{ x_1, x_6,x_7 \}
					\end{aligned}$
					}\\
					\\
					\end{tabular}
					}
				child{
					node{
						\begin{tabular}{r c}
						$u_6$:    & \multirow{2}{*}{
						$\begin{aligned}
						\lambda &= \{R_1, R_7 \}  \\
						\chi &= \{ x_1, x_7,x_8 \}
						\end{aligned}$
							}\\
						\\
						\end{tabular}
						}
					child{
						node{
							\begin{tabular}{r c}
							$u_7$:    & \multirow{2}{*}{
							$\begin{aligned}
							\lambda &= \{R_1, R_8 \}  \\
							\chi &=  \{ x_1, x_8,x_9 \}
							\end{aligned}$
							}\\
							\\
							\end{tabular}
						}
						child{
							node{
								\begin{tabular}{r c}
								$u_8$:    & \multirow{2}{*}{
								$\begin{aligned}
								\lambda &= \{R_1,R_9 \} \\
								\chi &=  \{ x_1, x_9,x_{10} \}
								\end{aligned}$
								}\\
								\\
								\end{tabular}
							}
						}
					}
				}
			}
		}
	}
};
\end{tikzpicture}

\caption{HD $\mathcal{D}$ of hypergraph $H$ from Section \ref{sect:beispiel}}
\label{fig:beispiel}
\end{subfigure}
\hspace{5mm}
\begin{subfigure}[b]{0.3\textwidth}
\centering
\begin{tikzpicture}
[sibling distance=10em,level distance=4em,
every node/.style={shape=rectangle,draw,align=center}]
\node{
		$\begin{aligned}
		\lambda &= \{R_1, R_7 \}  \\
		\chi &=  \{ x_1, x_7,x_8 \}
		\end{aligned}$		
	}
	child{
		node{
				$\begin{aligned}
				\lambda &= \{R_1, R_8 \}  \\
				\chi &=  \{ x_1, x_8,x_9 \}
				\end{aligned}$		
			}
			child{
				node{
					$\begin{aligned}
					\lambda &= \{R_1, R_9 \}  \\
					\chi &=  \{ x_1, x_9,x_{10} \}
					\end{aligned}$				
				}
			}
	};
\end{tikzpicture}

\caption{HD-fragment $\mathcal{D}_{1.1}$ implicitly constructed by Call 1.1 of function \Decomp}
\label{fig:callOneOne}
\end{subfigure}
\hspace{5mm}
\begin{subfigure}[b]{0.3\textwidth}
\centering

\begin{tikzpicture}
[sibling distance=10em,level distance=4em,
every node/.style={shape=rectangle,draw,align=center}]
\node{
	$\begin{aligned}
	\lambda &= \{R_1, R_3 \} \\
	\chi &=  \{ x_1, x_3,x_4 \}
	\end{aligned}$
}	
	child{
		node{
		$\begin{aligned}
		\lambda &= \{R_1, R_4 \} \\
		\chi &=  \{ x_1, x_4,x_5 \}
		\end{aligned}$
	}	
	child{
		node{
			$\begin{aligned}
			\lambda &= \{R_1, R_5 \}  \\
			\chi &=  \{ x_1, x_5,x_6 \}
			\end{aligned}$		
		}
		child{
			node{
				$\begin{aligned}
				\lambda &= \{ s_1\}  \\
				\chi &=  \{ x_1, x_6,x_7 \}
				\end{aligned}$				
			}
		}
	}
};
\end{tikzpicture}

\caption{HD-fragment $\mathcal{D}_{1.2}$  implicitly constructed by Call 1.2 of function \Decomp}
\label{fig:callOneTwo}
\end{subfigure}

\caption{Visualisations of the HD constructed as part of Section~\ref{sect:beispiel} and the HD-fragments used for its construction.}
\end{figure*}
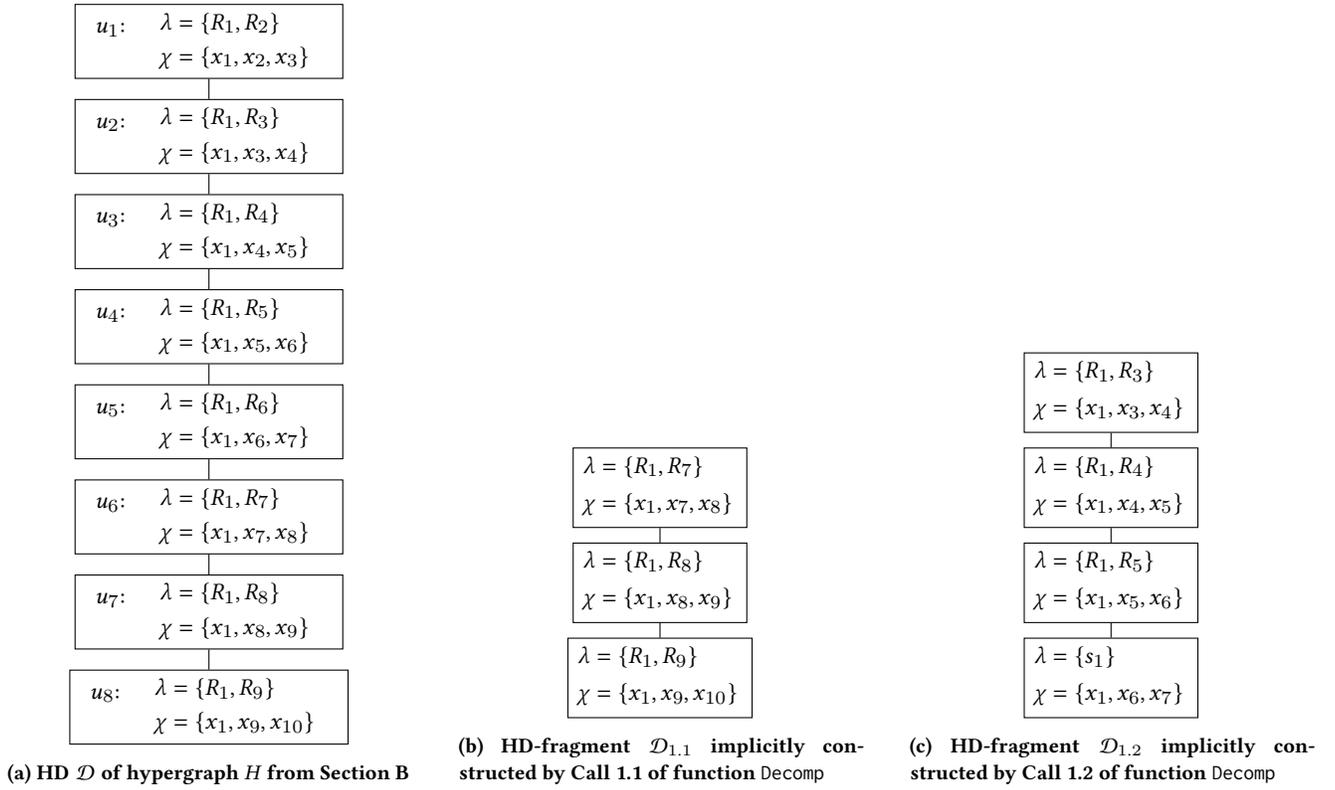

We now walk through algorithm \logk, which will allow us to 
see also  the notions from Section~\ref{sect:basics} in action.
Suppose that we run \logk with 
hypergraph $H$ and parameter $k = 2$.
In each of the loops on lines 3, 16, and 24, the algorithm searches for a $\lambda$-label until it finds a successful one. By ``succsesful'' we mean that the current execution of the main program or of function \Decomp returns true (on line 9 or 39, respectively). To keep things simple
in our discussion below, 
we will directly choose a successful one with the understanding, that this 
particular $\lambda$-label will eventually be selected by the program unless another successful one has already been found before.

\medskip

\noindent
{\em Main program.}
The RootLoop tests all possible candidates for $\lambda_r$. 
It will eventually try
$\lambda_r = \{R_1, R_2\}$ (as mentioned above -- with the understanding that 
it has not found another successful one before). 
There is only one $[\lambda_r]$-component $y$, namely 
$y.E = \{R_3, \dots, R_{10}\}$ and $y.Sp = \emptyset$.
Moreover, on line 6, we set $Conn_y = \{x_1, x_3\}$
and call function \Decomp for this combination of 
component $y$ and vertex set $Conn_y$.

\medskip

\noindent
{\em Call 1 of function \Decomp}
with parameters $H'.E = 
\{R_3, \dots, R_{10}\}$, 
$H'.Sp = \emptyset$, and
$Conn = \{x_1, x_3\}$.
Since none of the conditions of the base case is satisfied,
the ParentLoop will be entered. 
It will eventually try
$\lambda_p = \{R_1, R_5\}$. 
This splits $H'$ into 2 components
$c_1 = \{R_3,R_4\}$ and 
$c_2 = \{R_6,R_7,R_8,R_9,R_{10}\}$.
Component $c_2$ satisfies the size constraint on line 18 and,
therefore, becomes $comp_{down}$ on line 19.

The ChildLoop will eventually try $\lambda_c =  \{R_1, R_6\}$.
On line 25, we thus set $\chi_c = \{x_1, x_6, x_7\}$.
This gives rise to a single component 
$comps_p[i]$
inside $comp_{down}$, namely 
$comps_p[i] = \{R_7,R_8,R_9,R_{10}\}$. 
This component satisfies both the 
constraint for connectedness (line 26) and the size constraint (line 29). 
Moreover, it leads to 2 recursive calls of function \Decomp:
one for the component 
$comps_p[i] = \{R_7,R_8,R_9,R_{10}\}$
``below'' the ``child node'' (on line 33) and one for the component 
$c_1 = \{R_3,R_4,R_5\}$  ``above'' (on line 37). 
For the component ``below'', we have
$x.E = \{R_7,R_8,R_9,R_{10}\}$ and 
$x.Sp = \emptyset$; moreover, we 
set $Conn_x = \{x_1, x_7\}$ on line 32. For the component ``above'',
we set $comp_{up}.E = \{R_3,R_4,R_5\}$ on line 35 and 
$comp_{up}.Sp = \{s_1\}$ with $s_1 = \chi_c = \{x_1, x_6,x_7\}$ on line 36;
moreover, we take $Conn  = \{x_1, x_3\}$ 
from the current call of \Decomp.
Clearly, edge $R_6$ from $H'$ is already 
covered by $\chi_c = \{x_1, x_6, x_7\}$ 
and does not need to be further considered.

As we shall work out next,  Call~1.1 of function \Decomp
for the component ``below'' will return true 
based on the HD-fragment  $\mathcal{D}_{1.1}$  shown in 
Figure~\ref{fig:callOneOne}. Likewise, 
Call~1.2 of function \Decomp
for the component ``above'' will return true 
based on the HD-fragment  $\mathcal{D}_{1.2}$  shown in 
Figure~\ref{fig:callOneTwo}. The leaf node of 
$\mathcal{D}_{1.2}$ contains the  special edge $s_1$, which 
acts as a placeholder for the node $c$ with
labels 
$\lambda_c =  \{R_1, R_6\}$ 
and $\chi_c = \{x_1, x_6, x_7\}$ 
from the current call of \Decomp. 

The HD-fragment $\mathcal{D}_{1}$  of the successful Call~1
of function \Decomp is then obtained 
by taking  HD-fragment $\mathcal{D}_{1.2}$, 
replacing the leaf node with $\lambda$-label 
$\{s_1\}$ by the node $c$ with 
$\lambda_c = \{R_1,R_6\}$ 
and $\chi_c = \{x_1, x_6, x_7\}$ and
appending the HD-fragment $\mathcal{D}_{1.1}$
below this node $c$.
In other words, 
the HD-fragment $\mathcal{D}_{1}$ 
is
precisely 
HD $\mathcal{D}$ minus its root node, i.e., 
nodes $u_2$ -- $u_8$ of $\mathcal{D}$. 
The final HD $\mathcal{D}$ 
(shown in see Figure \ref{fig:beispiel})
is obtained by combining the root node
(whose $\lambda$-label was determined in the RootLoop of the 
main program, line 3) with 
HD-fragment $\mathcal{D}_{1}$.

\medskip

\noindent
{\em Call 1.1 of function \Decomp}
with parameters $H'.E = \{R_7,R_8,R_9,R_{10}\}$,
$H'.Sp = \emptyset$, and $Conn = \{x_1, x_7\}$. 
The ParentLoop (line 16) and the ChildLoop 
(line 24) will 
eventually choose $\lambda_p = \{ R_1,R_7\}$
and $\lambda_c = \{R_1,R_8\}$, respectively. Then function \Decomp is 
called recursively with parameters 
$x.E = \{R_9,R_{10}\}$, $x.Sp = \emptyset$, 
and $Conn_x = \{x_1,x_9\}$ on line 33 for the only component ``below''
and with parameters 
$comp_{up}.E = \{R_7\}$, 
$comp_{up}.Sp = \{s_2\}$ with $s_2 = \{x_1, x_8,x_9\}$,
and $Conn = \{x_1,x_7\}$ on line 37 for the component ``above''.

\medskip

\noindent
{\em Call 1.1.1 of function \Decomp}
with parameters 
$H'.E = \{R_9,R_{10}\}$, $H'.Sp = \emptyset$, 
and $Conn = \{x_1,x_9\}$.
This call immediately returns true 
since we have reached the base case on lines 12 - 13. The corresponding 
HD-fragment $\mathcal{D}_{1.1.1}$
consists of a single node with 
$\lambda$-label $\{R_9, R_{10}\}$.

\medskip

\noindent
{\em Call 1.1.2 of function \Decomp}
with parameters 
$H'.E = \{R_7\}$, 
$H'.Sp = \{s_2\}$ with $s_2 = \{x_1, x_8,x_9\}$,
and $Conn = \{x_1,x_7\}$.
In the ParentLoop, eventually, 
$\lambda_p = \{ R_1,R_6\}$ 
will be chosen on line 16.
In this case, 
$comp_{down}$ is actually all of $H'$ and
it clearly satisfies the size constraint on line 18.
In the ChildLoop, 
$\lambda_c = \{R_1,R_7\}$ will eventually be chosen
on line 24.
It gives rise to $\chi_c = \{x_1, x_7, x_8\}$
on line 25 with a single $[\chi_c]$-component
$comps_c[i] = \{s_2\}$.
Function  \Decomp 
will therefore be called 
on line 33 with parameters 
$x.E = \emptyset$, 
$x.Sp = \{s_2\}$, 
and $Conn_x = \{x_1,x_8\}$. 
This call (referred to as Call 1.1.2.1) 
returns true since 
we now have the base case on lines 14 - 15.

The recursive call of function \Decomp on 
line 37 for the ``components above'' 
is a special case where there are actually no such ``components above'' left. Hence, in 
this case, we have
$comp_{up}.E = \emptyset$, 
$comp_{up}.Sp = \{s_3\}$
with $s_3 = \chi_c = \{x_1, x_7, x_8\}$, and 
$Conn = \{x_1,x_7\}$.
This leads to the Call 1.1.2.2 of function \Decomp, 
which returns true since 
we again have the base case on lines 14 - 15.

In total, Call~1.1.2 of function \Decomp
is successful and the corresponding 
HD-fragment $\mathcal{D}_{1.1.2}$
consists of 2 nodes: the root node with 
$\lambda$-label $\{R_1,R_7\}$ and its child
node with $\lambda$-label $\{s_2\}$

We can now also construct the 
HD-fragment $\mathcal{D}_{1.1}$
of the successful Call~1.1 of function \Decomp.
More precisely, HD-fragment $\mathcal{D}_{1.1}$ is obtained
by taking  HD-fragment $\mathcal{D}_{1.1.2}$, 
replacing the leaf node with $\lambda$-label 
$\{s_2\}$ by the node $c$ with 
$\lambda_c = \{R_1,R_8\}$ 
and $\chi_c = \{x_1, x_8, x_9\}$ from Call~1.1 and
appending the HD-fragment $\mathcal{D}_{1.1.1}$
below this node $c$.
The resulting 
HD-fragment $\mathcal{D}_{1.1}$
is shown in Figure~\ref{fig:callOneOne}. 
That is, $\mathcal{D}_{1.1}$ is the 
subtree consisting of the bottom 3 nodes 
$u_6$, $u_7$, and $u_8$, 
of the final HD $\mathcal{D}$ displayed in 
Figure \ref{fig:beispiel}.

\medskip

\noindent
{\em Call 1.2 of function \Decomp}
with parameters 
$H'.E = \{R_3, R_4, R_5\}$, 
$H'.Sp = \{s_1\}$ with $s_1 = \{x_1, x_6,x_7\}$,
and $Conn = \{x_1, x_3\}$.
The execution of this function call is very similar to the calls discussed above. Below, we therefore 
do not discuss in detail the remaining recursive calls inside Call~1.2. Instead, we only list for each such call the parameters, the balanced separator $\chi_c$, and the corresponding HD-fragments.

In Call 1.2 of function \Decomp, eventually the 
balanced separator with $\lambda_c = \{R_1, R_4\}$ 
and $\chi_c = \{x_1,x_4, x_5\}$ 
will be chosen, which gives rise to the 
recursive Calls~1.2.1 (line 33) and 
1.2.2 (line 37) of function \Decomp, which 
we briefly discuss below.

\medskip

\noindent
{\em Call 1.2.1 of function \Decomp}
with parameters 
$H'.E = \{R_5\}$, 
$H'.Sp = \{s_1\}$ with $s_1 = \{x_1, x_6,x_7\}$,
and $Conn = \{x_1, x_5\}$.
As in the Call~1.1.2, we now have an extended
subhypergraph of $H$ consisting of a single 
edge and a single special edge. 
Analogously to Call~1.1.2, also 
Call~1.2.1 returns true and the
corresponding 
HD-fragment $\mathcal{D}_{1.2.1}$
consists of 2 nodes: the root node with 
$\lambda$-label $\{R_1,R_5\}$ and its child
node with $\lambda$-label $\{s_1\}$.

\medskip

\noindent
{\em Call 1.2.2 of function \Decomp}
with parameters 
$H'.E = \{R_3\}$, 
$H'.Sp = \{s_4\}$ with $s_4 = \{x_1, x_4,x_5\}$,
and $Conn = \{x_1, x_3\}$.
Again, we have an extended
subhypergraph of $H$ consisting of a single 
edge and a single special edge. 
Analogously to the Calls~1.1.2 and 1.2.1, 
also Call~1.2.2 returns true and the
corresponding 
HD-fragment $\mathcal{D}_{1.2.2}$
consists of 2 nodes: the root node with 
$\lambda$-label $\{R_1,R_3\}$ and its child
node with $\lambda$-label $\{s_4\}$.

We can now construct the 
HD-fragment $\mathcal{D}_{1.2}$
of the successful Call~1.2 
by taking  HD-fragment $\mathcal{D}_{1.2.2}$, 
replacing the leaf node with $\lambda$-label 
$\{s_4\}$ by the node $c$ with 
$\lambda_c = \{R_1,R_4\}$ 
and $\chi_c = \{x_1, x_4, x_5\}$ from Call~1.2 and
appending the HD-fragment $\mathcal{D}_{1.2.1}$
below this node $c$.
The resulting 
HD-fragment $\mathcal{D}_{1.2}$
is shown in Figure~\ref{fig:callOneTwo}.

\section{Further Combinatorial Observations and Optimisations}
\label{sect:improvements}

As was shown in Theorem \ref{theo:recursion-depth},
algorithm \logk introduced in Section~\ref{sect:Algorithm:ShortForm}
reaches the primary goal of splitting the HD-construction
into subtasks with guaranteed upper bound on their size. In theory, this is
enough to support parallelism. However, this basic algorithm still leaves a lot of room for further improvements.
In this section, we present several optimisations, which are crucial to achieve good performance in practice. 
The line numbers below refer to Algorithm~\ref{alg:parHD}. 
However, in~Algorithm~\ref{alg:parHDOpt}, 
we will ultimately also give the pseudo-code for the  enhanced algorithm where all the optimisations mentioned below are included.

\smallskip

\noindent
{\bf Extension of the base case.} The recursive function \Decomp starts (on lines 12 -- 15)  
with some simple checks that immediately give a ``true'' answer. In contrast, a ``false'' answer
is only obtained in case of unsuccessful execution of the entire procedure. We could 
add the following negative case to the top of the procedure:  if $H'.E = \emptyset$, then $|H'.\Sp| \leq 1$ must hold.
The rationale of this condition is that,  if there are no more edges in $H'.E$, then 
we would have to use only  ``old'' edges (i.e., edges covered already at some 
node further up in the HD) in the $\lambda$-label to separate the remaining 
special edges. However, a $\lambda$-label consisting of ``old'' edges only is not allowed, 
since this would violate the second condition of the normal form in Definition~\ref{def:normalform}
(i.e., "some progress has to be made").

\smallskip

\noindent
{\bf Root of the HD-fragment.} In the current form of procedure \Decomp,  we always ``guess''
a pair $(p,c)$ of nodes, such that $p$ is the parent of $c$.
This also covers the case that $c$ is the root node
of the HD-fragment for the current extended subhypergraph. In this case, the
parent node $p$ would actually be the node immediately above this HD-fragment 
(in other words, $p$ was the node from which the current
call of \Decomp happened). However, it would be more efficient to consider the case of ``guessing'' the root
node of this HD-fragment explicitly. More precisely, we would thus first check for the
label $\lambda_p$ guessed in \Decomp on line 16 (which, in the current version of the algorithm, is automatically
treated as the ``parent'') if all $[\lambda_p]$-components have at most half the size of the current
extended subhypergraph.  

If this is the case, then we may use this node as the root of the HD-fragment to cover the current extended
subhypergraph. This makes sense since it corresponds precisely to the ``search'' for a balanced
separator in the proof of Lemma~\ref{lem:existence-balsep}. That is, if the root of the HD gives
rise to components which are all at most half the size, then the root {\em is\/} the desired
balanced separator.

If this is not the case, then we simply proceed with procedure \Decomp in its present form, i.e.: 
there exists exactly one $[\lambda_p]$-com\-po\-nent whose size is
bigger than half. So we take the guessed node as the parent and search 
for a balanced separator as a child of $p$ in the direction of this 
oversized $[\lambda_p]$-component.

\smallskip

\noindent
{\bf Allowed edges.} The main task of procedure \Decomp is to compute labels (i.e., edge sets) 
$\lambda_p$ and $\lambda_c$ of nodes $p,c$, which will ultimately be in a parent-child relationship
in the HD. 
For these labels, Algorithm~\ref{alg:parHD} imposes no restriction. That is, 
in principle, we would try all possible sets of $\leq k$ edges for these labels. However, not all edges 
actually make sense. We should thus add one more parameter to procedure \Decomp indicating
the edges that are allowed in a $\lambda$-label of the HD-fragment for this extended
subhypergraph. 

More specifically, in our search for the $\lambda$-label of some node $u$, 
we may exclude from the HD of 
the extended subhypergraph $comp_{\up}$ (i.e., in the recursive call of function 
$\Decomp$ on line 37) all edges
which are part of some component ``below'' $u$.  The rationale of this restriction is that, 
by the special condition, using a ``new'' edge in a $\lambda$-label 
forces us to add all its vertices to the $\chi$-label, i.e.: it is fully covered in such a node. But then it cannot be
part of a component whose edges are covered for the first time further down in the tree.

Note that we can yet further restrict the search for the label $\lambda_c$ by requiring that at least one
edge must be from $H'.E$, since choosing only ``old'' edges would violate the second condition of the normal form.
As far as the label $\lambda_p$ is concerned, the same kind of restriction can be applied if we first implement 
the previous optimisation of handling the root node of the current HD-fragment separately. 
If we indeed have to guess the labels $\lambda_p$ and 
$\lambda_c$ of {\em two\/} nodes $p$ and $c$ 
(i.e., the label $\lambda_p$ guessed first was not a balanced separator), 
then both nodes $p$ and $c$ are {\em inside\/} the current HD-fragment. 
Hence, also the label $\lambda_p$ must contain at least one ``new'' edge.

\smallskip

\noindent
{\bf No special treatment of the root of the HD.} In the current form of the algorithm,
 we start in the main program by ``guessing'' the label $\lambda_r$ of the root of the HD on line 3 and then branch into 
 calls of procedure \Decomp  for each $[\lambda_r]$-component. Of course, there is no guarantee that this 
 $\lambda$-label is a balanced separator. 
 Consequently, there is no guarantee that the size of {\em all\/} HD-fragments to be constructed in these calls of 
 procedure \Decomp is significantly smaller than the entire HD. 
 
 In order to start with a balanced separator right from the beginning, 
 we may instead call procedure \Decomp straight away with parameters $H'.E = H$, $H'.\Sp = \emptyset$, and
 $\Con = \emptyset$. We  thus treat the search for the very first $\lambda$-label in the desired HD in 
 exactly the same way as for any other HD-fragment. As far as the above mentioned optimisation of 
 restricting the allowed edges is concerned, of course all edges of $H$ would be initially allowed.
 
\smallskip

\noindent
{\bf Searching for  child nodes first.} In Algorithm~\ref{alg:parHD}, we first look for
 $\lambda$-labels of potential parent nodes, and consider afterwards the $\lambda$-labels
 of potential child nodes. Only then do we check if the $\chi$-label of the child is a balanced
 separator of the current subcomponent. We have observed that in many hypergraphs of HyperBench, balanced
 separators are rare, in the sense that only a small part of the search space will ever fulfil the
 properties required. 
 Therefore we should first look for a potential child s.t. its $\lambda$-label is a balanced
 separator, and only afterwards try to find a fitting parent. While this may seem slightly 
 unintuitive, it allows us to quickly detect cases where no balanced separator can be found at all. 
 
 Note that we can determine the precise bag $\chi_c$ for a child $c$ only when we know the $\lambda$-label
 of its parent. Nevertheless, even if we only have $\lambda_c$, we can over-approximate the $\chi_c$-label 
 as $\bigcup \lambda_c$. Hence, if $\bigcup \lambda_c$ is not a balanced separator, 
 then we may clearly conclude that  neither is $\chi_c$.

Finally, note that by searching for the child node first, we get the above 
described optimisation of treating the ``Root of the HD-fragment'' separately almost for free. Indeed, when 
computing $\lambda_c$, we can immediately check if $Conn \subseteq \bigcup \lambda_c$ holds. 
Recall that $Conn$ constitutes the interface to the HD-fragment {\em above\/} the current one. Hence, 
if $\bigcup \lambda_c$ fully covers this interface, $c$ is in fact the root node of the current HD-fragment.

\smallskip

\noindent
{\bf Speeding up the search for parent $\mathbf{\lambda}$-labels.} The previous optimisation means that, 
after having found a $\lambda$-label $\lambda_c$ for the child which is a balanced separator of the current 
subcomponent, we need to find a {\em suitable\/} $\lambda$-label of the parent.  
By ``suitable'' we mean that 
we may limit ourselves to edges which have a non-empty intersection 
with $\bigcup \lambda_c$. A very high-level explanation why we may exclude edges $e$ 
with $e \cap \bigcup \lambda_c = \emptyset$
from the 
search space of $\lambda_p$ is that the control flow of  function \Decomp is 
mainly determined by the edges and special edges covered by $\bigcup \lambda_c$ and 
the $[\lambda_c]$-components {\em below} $c$. By the connectedness condition, 
if $e$ is covered {\em above\/} $c$ and has empty intersection with $\bigcup \lambda_c$, 
then excluding or including $e$ in $\lambda_p$ has no 
effect on the $[\lambda_c]$-components {\em below} $c$.
In our experimental evaluation, we found that this restriction indeed significantly reduces the 
time it takes to 
either find a {\em suitable\/} $\lambda_p$, or detect that no such $\lambda$-label exists. 
Of course, this restriction of the search space cannot destroy soundness.  We will show below that 
also the completeness of the algorithm is preserved.

\begin{theorem}
\label{theo:algorithm-improved}
The optimised \logk algorithm for checking if a hypergraph $H$ has $\hw(H) \leq k$ 
given in 
Algorithm~\ref{alg:parHDOpt}
is sound and complete. More specifically, 
for given hypergraph $H$ and integer $k \geq 1$, 
the algorithm  returns ``true'' if and only if there exists an 
HD of $H$ of width $\leq k$. 
Moreover, by materialising the decompositions implicitly constructed in the recursive calls of the 
\Decomp function, an HD of $H$ of width $\leq k$ can be constructed in polynomial time 
in case of a successful computation (i.e., return-value ``true'').
\end{theorem}

\begin{proof}
The soundness and completeness of Algorithm~\ref{alg:parHDOpt} follow almost immediately from the 
soundness and completeness of Algorithm~\ref{alg:parHD} together with the above explanations of the various optimisations. 
Likewise, the polynomial-time upper bound on the time needed to construct an HD in case of a successful computation
can again be easily shown as part of the soundness proof. 
The only non-trivial part is that the last optimisation (i.e., the restriction of the search space for 
$\lambda(p)$) does not destroy the completeness of the algorithm. 
The remainder of the proof will concentrate on this aspect.

Assume that hypergraph $H$ has an HD of width $\leq k$. Then the optimised \logk algorithm without the 
restriction on the search space for label $\lambda_p$ (on line 22) returns the overall result {\em true\/}. 
This is due to the fact that,
as was argued in Section \ref{sect:improvements}, the other optimisations mentioned there do not affect 
the completeness of the algorithm.
Now consider a recursive call of function \Decomp and suppose that it returns true if the restriction 
on the search space for label $\lambda_p$ is dropped.
Of  course, if the value {\em true\/} is returned in one of the base cases (lines 6 or 8) or 
if $\lambda_c$ turns out to be the $\lambda$-label of the root node of the current HD-fragment (and {\em true\/}
is returned on line 21), then the restriction of the search space for $\lambda_p$ has no effect at all. 
Hence, the only interesting case to consider is that the parent loop (lines 22 -- 43) is indeed executed. 

Let $\lambda_p$ be the $\lambda$-label chosen on line 22 if no restriction is imposed on the search space.
We claim that we may remove from $\lambda_p$ all edges that have an empty intersection with 
$\bigcup \lambda_c$ without altering the control flow of this particular execution of function $\Decomp$.
Actually, it suffices to show that we may remove {\em one\/} edge $e$ with 
an empty intersection with $\bigcup \lambda_c$ from $\lambda_p$
without altering the control flow of this particular execution of function $\Decomp$.
Then the claim follows by an easy induction argument.

So suppose that $\lambda_p$ contains at least one edge $e$ such that $e \cap \bigcup \lambda_c = \emptyset$
and let $\lambda'_p = \lambda_p \setminus \{e\}$. 
An inspection of the code of the parent loop reveals that it suffices to show that
this elimination of edge $e$ from $\lambda_p$ leaves $comp_{\low}$ unchanged.
Indeed, if $comp_{\low}$ is still a $[\lambda'_p]$-component, say the $i$-th $[\lambda'_p]$-component, 
then the if-condition on line 24 is true. Of course, there can be only one $[\lambda'_p]$-component satisfying the
condition $comps_p[i] |  >  \frac{|H'|}{2}$. 
Hence, on line 25, for this particular $i$, exactly the same value is assigned to $comp_{\low}$ for $\lambda'_p$ as for 
$\lambda_p$. But then also $\chi_c$ on line 28 gets the same value as without the restriction on the search 
space of $\lambda_p$. Consequently, also the $[\chi_c]$-components computed on line 33 and the 
parameters supplied to the recursive calls of function \Decomp (on lines 36 and 41) remain the same  
as without the restriction on the search space. Hence, function \Decomp will ultimately return the value {\em true\/}
also if we choose $\lambda'_p$ on line 22.

It remains to show that $\lambda_p$  and $\lambda'_p$ indeed give rise to the same 
component $comp_{\low}$. To avoid confusion, let us write $comp_{\low}$ to denote 
a $[\lambda_p]$-component and $comp'_{\low}$ to denote a $[\lambda'_p]$-com\-ponent.
Let $comp_{\low}$ be the unique $[\lambda_p]$-component 
that satisfies the condition $| comps_p[i] |  >  \frac{|H'\!|}{2}$ on line 24.
We have $\lambda'_p \subseteq \lambda_p$. Decreasing a set can only increase the corresponding 
components. Hence, there exists a $[\lambda'_p]$-component, call it $comp'_{\low}$ with 
$comp_{\low} \subseteq comp'_{\low}$. 
We have to show that $comp_{\low} = comp'_{\low}$ holds.

The set $comp_{\low}$ consists of  the edges and special edges 
of the $[\chi_c]$-components contained in $comp_{\low}$ (denoted as $comps_c$ in the algorithm), 
and the edges and special edges covered by $\chi_c$. Let us refer to these $[\chi_c]$-components as
$C_1, \dots, C_\ell$. By Corollary~\ref{cor:chi-vs-lambda}, these $[\chi(c)]$-components
are at the same time the $[\lambda_c]$\-components contained in $comp_{\low}$.
And the edges and special edges covered by $\chi_c$ are of course also 
covered by~$\bigcup \lambda_c$.
Likewise, 
$comp'_{\low}$ consists of the (special) edges 
of the $[\lambda_c]$\-components contained in $comp'_{\low}$ 
plus the (special) edges covered by $\lambda_c$. 

By $comp_{\low} \subseteq comp'_{\low}$, all $[\lambda_c]$-components $C_1, \dots, C_\ell$
contained in $comp_{\low}$ 
are of course also contained in $comp'_{\low}$. We have to show that there is no further  
$[\lambda_c]$-component contained in $comp'_{\low}$.
Assume to the contrary that there exists a $[\lambda_c]$-component $C'$ in $comp'_{\low}$
such that $C'$ is not in $comp_{\low}$. By definition, $comp_{\low}$ is $[\lambda_p]$\-connected
while $comp'_{\low}$ is $[\lambda'_p]$\-connected. Hence, there exist (possibly special) edges 
$f' \in C'$ and $f \in C_i$ for some $i \in \{1, \dots, \ell\}$, such that there is a 
path $\pi$ (represented as a sequence of edges) with 
$\pi = (f_0, f_1, \dots, f_m)$, such that 
$f = f_0$, $f' = f_m$, and 
$(f_\alpha \cap f_{\alpha +1}) \setminus \bigcup \lambda'_p \neq \emptyset$ for every 
$\alpha \in \{0, \dots, m-1\}$. W.l.o.g., choose $f$, $f'$, and $\pi$ such that $m$ is minimal. 
Since $f$ and $f'$ are not $[\lambda_p]$-connected,  there exists $\alpha$ with 
$f_\alpha \cap f_{\alpha +1} \cap e \neq \emptyset$ while 
$(f_\alpha \cap f_{\alpha +1}) \setminus \bigcup \lambda_p = \emptyset$.

Since all (special) edges in  $comp'_{\low}$ are either in some 
$[\lambda_c]$-component contained in $comp'_{\low}$ 
or covered by $\bigcup \lambda_c$, 
and since we are assuming that $\pi$ is of minimal length, 
the path $\pi$ starts with $f$ in some 
$[\lambda_c]$-component $C_i$, possibly goes through $\bigcup \lambda_c$ and 
ends with $f'$ in component $C'$. Recall that $e$ was chosen such that $e \cap \bigcup \lambda_c = \emptyset$. 
Hence, the edges $f_\alpha $ and $f_{\alpha +1}$ cannot be covered by $\bigcup \lambda_c$. 
By our assumption that $\pi$ has minimal length, we can also exclude the case that both 
$f_\alpha $ and $f_{\alpha +1}$ are in $C'$. Hence, at least one of $f_\alpha $ and $f_{\alpha +1}$ must be 
in $C_i$. In other words, $e \cap C_i \neq \emptyset$. Hence, also $e \cap comp_{\low} \neq \emptyset$.
However, by the check on line 31 in Algorithm~\ref{alg:parHDOpt}, we know that
$V(comp_{\low}) \cap \bigcup \lambda_p \subseteq \bigcup \lambda_c$.  
This contradicts the assumption that $e \in \lambda_p$  and $e \cap \bigcup \lambda_c = \emptyset$.
 \end{proof}

\begin{algorithm}
\setstretch{0.92}
\SetInd{0.1em}{0.8em}
\DontPrintSemicolon
\SetKwInOut{KwPara}{Parameter}
\SetKwInput{KwType}{Type}

\KwType{Comp=($E$: Edge set, $\Sp$: Special Edge set)}

\KwIn{$H$: Hypergraph}
\KwPara{$k$: width parameter}
\KwOut{\textbf{true} if \emph{hw} of $H$ $\leq k$, else \textbf{false}}

  \SetKwFunction{algo}{Decomp}  

  \Begin{
  $H_{comp} \coloneqq$ Comp($E$: $H$, $\Sp$: $\emptyset$) \\
    \textbf{return} \algo{$H_{comp}$, $\emptyset$, $H$} \Comment{initial call} 
    }

  \SetKwProg{myalg}{function}{}{}
    \myalg{\algo{$H'\!$: Comp, $Conn$: Vertex set,
    \hspace*{0em}~$A$:~{Edge~set}}} {
     \uIf(\Comment{\textbf{Base Cases}})
              {$|\text{$H'\!\!.E$}| \leq k$ \textbf{\emph{and}} $|H'\!\!.\Sp| = 0$} {
              \textbf{return true} \;
     }
     \uElseIf {$|H'\!\!.E| = 0 $ \textbf{\emph{and}} $|H'\!\!.\Sp| = 1$} {
    \textbf{return true} \;
     }     
     \ElseIf {$|H'\!\!.E| = 0 $ \textbf{\emph{and}} $|H'\!\!.\Sp| > 1$} {
    \textbf{return false} \;
     }

    \ForEach(\Comment{\textbf{ChildLoop}})
              { $\lambda_c \subseteq A $\! { s.t.\,}    $\lambda_c \cap H'.E \neq \emptyset$ \\ \nonl \tabto{2.5cm} $\text{\ \bf and } 1 \leq |\lambda_c| \leq k$}   { 

             $comps_c \coloneqq $ $[\lambda_c ]$-components of $H'$ \;

             \uIf{$\exists i $ s.t. $| comps_c[i] |  >  \frac{|H'|}{2}$} {  
                    \textbf{continue ChildLoop}
             }                  
             \ElseIf(\Comment{check if $\lambda_c$ is root}) 
             { $Conn \subseteq \bigcup  \lambda_c $} { 
                      $\chi_c \coloneqq  \bigcup \lambda_c \cap V(H')  $ \; 
                     \ForEach{$y \in comps_c$} {
                            $Conn_y \coloneqq V(y) \cap \chi_c$ \; 
                            \If {\textbf{not}(\algo{$y$, $Conn_y$, $A$})} {
                                     \textbf{continue ChildLoop} 
                            }
                     }

                \textbf{return true}   \Comment{$c$ is root of $H'$}
             }

             \label{asd} \ForEach(\Comment{\textbf{ParentLoop}})
              { $\lambda_p \subseteq A $\! { s.t.}  $\lambda_p \cap H'.E \neq \emptyset$ \\ \nonl \tabto{2.5cm} $\text{\ \bf and } 1\leq |\lambda_p| \leq k$}  {
                    $comps_p \coloneqq $  $[\lambda_p]$-components of  $H'$  \;

                     \uIf{$\exists i $ s.t. $| comps_p[i] |  >  \frac{|H'\!|}{2}$}    {
                            {$comp_{\low}\coloneqq comps_p[i]$ \Comment{found child comp.}} 
                     }
                     \Else{ 
                            \textbf{continue ParentLoop}
                     }  
                    $\chi_c \coloneqq   \bigcup \lambda_c \cap V(comp_{\low}) $\; 
                    \If{$V(comp_{\low}) \cap Conn \not \subseteq \bigcup \lambda_p $   }  {
                           \textbf{continue ParentLoop} \Comment{connect. check} 
                    }  
                    \If{$V(comp_{\low}) \cap \bigcup\lambda_p \not \subseteq  \chi_c $   }  {
                           \textbf{continue ParentLoop} \Comment{connect. check} 
                    }

                     $new\_comps_c \coloneqq $ $[\chi_c ]$-components of $comp_{\low}$ \;

                     \ForEach{$x \in new\_comps_c$} {
                            $Conn_x \coloneqq V(x) \cap \chi_c$ \;
                            \If {\textbf{not}(\algo{$x$, $Conn_x$, $A$})} {
                                   \textbf{continue ParentLoop}  \Comment{reject parent}
                            }
                     }

                     $comp_{\up} \coloneqq H'\! \setminus comp_{\low} $  \Comment{pointwise diff.} \\
                     $comp_{\up}.\Sp = comp_{\up}.\Sp \cup \{ \chi_c \} $              \;

                     $A_{\up} \coloneqq A \setminus comp_{\low}.E$ \Comment{reducing $A$} \\
                     \If{\textbf{not}(\algo{$comp_{\up}$, $Conn$, $A_{\up} $})} {
                            \textbf{continue ParentLoop} \Comment{reject parent} 
                     }
                     \textbf{return true} \Comment{$\hw$ of $H'\! \leq k$}                 
             }

    }
    \textbf{return false} \Comment{exhausted search space} 
       
}

\caption{Optimised \logk}
\label{alg:parHDOpt}
\end{algorithm}

\section{Additional Experimental Evaluation}
\label{sect:moreEval}
We provide here a number of additional details on our implementation, such as the hybridisation strategy employed in our implementation of \logk, as well as further experiments to highlight various properties of our contribution.

\subsection{Parallel Implementation}
\label{subsect:parallel}

For our experiments, we implemented \logk including all of the
optimisations presented in Section~\ref{sect:improvements}. As
discussed above, a crucial aspect of our algorithm design is that the
use of balanced separators allows us to recursively split the problem
into smaller subproblems. The subproblems are independent of each
other and are therefore processed in parallel by our
implementation. Furthermore, following observations made
in~\cite{DBLP:conf/ijcai/GottlobOP20}, our implementation also executes
the search for balanced separators in parallel by partitioning the
search space effectively.

\subsection{Hybrid Approaches}
\label{sec:hybrid}

While our algorithm has desirable properties for parallelisation (as has been pointed out in Section \ref{subsect:parallel} above), this
comes at the cost of some overhead when compared to simpler methods,
in particular \detk. Especially on small and simple instances the restriction to balanced separators may act as a detriment to performance that outweighs its benefits for parallelisation and its effect on severely restricting the search space.

To balance these considerations overall in practice we therefore also
consider hybrid variants of our implementation. Intuitively, we want to
use \logk as long as the subproblems are still complex, but once they
become simple, we want to switch to an  algorithm that is better suited for those cases. For the
simpler algorithm, \detk is the natural choice as it performs very
well on small instances as was shown in \cite{10.1145/3440015}, where an implementation of \detk is provided as part of \newdetk.
To determine when the switch is made, we implemented two simple metrics to capture the complexity of a hypergraph:
\begin{description}
\item[\edgecount]
  In \edgecount we simply use the number of edges of the hypergraph $|E(H)|$ as the measure of complexity.
\item[\relfill] The \relfill metric is characterised by
  the formula $|E(H)|\frac{k}{\mathit{avg}_{e\in E(H)}|e|}$ where $k$ is the width parameter of the algorithm. The additional factor compared to \edgecount is best understood as two separate additional weightings. Higher width implies more complex structure and hence we expect more complexity per edge. On the other hand, if edges are on average larger, then it becomes easier to find covers and we therefore also inversely weight by the average cardinality of the edges.
\end{description}

We investigated the effectiveness of these metrics through a series of experiments. In particular, we ran experiments with both metrics and different thresholds for when to switch from \logk to \detk. To be precise, for a metric $m$ and threshold $T$,  \logk is executed for a subproblem with hypergraph $H_i$ as long as $m(H_i) \geq T$. If $m(H_i) < T$, we switch to an implementation of \detk written from scratch as part of the code base of \logk. A similar strategy was already proposed in~\cite{DBLP:conf/ijcai/GottlobOP20} in the context of \balgo. However, in that system no metric for the complexity of a subproblem was employed, but rather the switch to \detk was always performed at a fixed recursion depth.

A significant portion of the hypergraphs in HyperBench are relatively
small, so that even the full problem would be too simple for our metrics
to be above reasonable thresholds. In those cases, the execution would
be equivalent to simply running \detk on the instance. We therefore
exclude small hypergraphs from these experiments to obtain more
meaningful results by considering only the \hblarge instances here.

Our results for the experiments on \hblarge are summarised in Table~\ref{tab:hybridNewTest}. Methods \relfill and \edgecount refer to \logk with the respective metric used for hybridisation. The threshold column refers to parameter $T$ in the discussion above. The experiments for all \logk hybrid methods had access to 12 cores, the experiments for \newdetk and \hdsmt used only 1 core each since they do not support parallelism.

\begin{table}[t]
  \caption{Study of two Hybrid methods of \logk on \hblarge, and a comparison with \newdetk and \hdsmt for reference.} 
  \label{tab:hybridNewTest}
  \centering
  \begin{tabular}{l|c|c|c}
    \toprule
  Method & Threshold & Solved  & Av. runtime (sec.)  \\
   \midrule
\relfill & 200 &395 & 92.15 \\
\relfill & 400 & \textbf{411}&  93.53 \\
\relfill & 600 & 410&  \textbf{87.86} \\
\edgecount & 20 & 171 &  130.0\phantom{0} \\
\edgecount & 40 & 219 &  145.09 \\
\edgecount & 80   & 292 &  117.33 \\
  \midrule
  \newdetk~\cite{10.1145/3440015} & - & 174 & 318.93 \\
  \hdsmt~\cite{DBLP:conf/ijcai/SchidlerS21} & - & 277 & 779.39 \\
\bottomrule
  \end{tabular}
\end{table}

Overall, \relfill clearly performs best, especially in the number of solved instances. For thresholds 400 and 600, approximately 90\% of the 465 large instances from \hblarge were solved. This constitutes a significant improvement over the 37\% and 60\% achieved by \newdetk and \hdsmt, respectively. Note that despite solving more instances -- for which \detk and \hdsmt timed out -- the running time is also at least 3 times lower for \relfill. This is surprising as we do not consider timed out instances in our average running time calculations.

One surprising observation from the table is that the differences in performance between different thresholds are much smaller for \relfill than for \edgecount.
Further investigation suggests that this is due to the \relfill metric decreasing much more rapidly as hypergraphs become simpler. 
At the same time, for subproblems that fall in the range between 200-600, the performance of switching to \detk immediately is roughly the same (on average) 
as the performance of continuing with \logk for one or two more steps.

While \edgecount performs worse than \relfill, we can still see a clear improvement over the state of the art methods \newdetk and \hdsmt. Especially the significant improvement over \detk is important to observe as it clearly demonstrates the benefits of our hybrid approach. Recall that when we split our problem into balanced subproblems, each subproblem is then solved independently in parallel. In the hybrid variant we will thus eventually execute our implementation of the \detk algorithm on multiple subproblems in parallel, i.e., we can use an inherently single-threaded algorithm effectively in parallel because we are able to create balanced subproblems. 

\subsection{\hdsmt with 10 Hour Timeout}
As mentioned before, the method used in  \hdsmt fundamentally differs from the search algorithm here. In \hdsmt, the problem is encoded to the SAT modulo theories (SMT) setting and the final solving step is handed off to standard SMT solvers. The encoding in \hdsmt is constructed in such a way that no width parameter is handed to the solver, rather the encoding will always return the optimal width as its solution. This is a significant difference to the parameterised search implemented in \logk (but also previous algorithms such as \detk and \balgo).

This difference makes it naturally difficult to compare running times and the number of optimal solutions directly. To provide a fuller picture we add here Table~\ref{tab:smalleval} to  provide a fuller picture. In the table we present the results from running the same experiments with \hdsmt but with the timeout increased to 10 hours. This increase in timeouts naturally makes the average runtimes difficult to compare to the values in Table~\ref{tab:bigeval}. However, importantly we see that the increase in solved instances is also moderate, and overall \logk still solves significantly more instances than \hdsmt with 10 hours of maximum running time.

\subsection{Analysis on Scaling by Size}

To gain further insight into the performance of each algorithm with respect to solving instances optimally, we investigate the size of solved and unsolved instances. To this end, Figure~\ref{fig:unsolvedComp} provides (logarithmic) scatter plots for each of the three algorithms 
in our tests.
In each plot, each instance is positioned according to its number of vertices and edges. Solved instances for each algorithm are drawn in green while unsolved instances are drawn in red.

The plots show that our intuition holds true in that solving large instances (in both axis) significantly benefits from using \logk. Most of the remaining hypergraphs are either extremely large, containing thousands of edges and vertices, or belong to very specific CSP classes which we know to have very high width (significantly beyond the width 10 limit used in our experiments) through graph-theoretic arguments. %

\definecolor{negcol}{HTML}{EF3340}
\definecolor{poscol}{HTML}{00AB84}

\begin{figure*}[t]
\centering
\begin{subfigure}[b]{0.33\textwidth}
  \centering
\begin{tikzpicture}[transform shape,scale=0.8]
  \begin{axis}[
    xmin=0, 
    xmode=log,
    ymode=log,
    xlabel={Number of edges},
    ylabel={Number of vertices},
  ]
    \addplot+[color=poscol,only marks,mark size=1pt,mark=+] table[x=edges, y=vertices, col sep=comma] 
    {data/solved_newdetk.csv};
    \addplot +[color=negcol,only marks,mark size=1pt,mark=x] table[x=edges, y=vertices, col sep=comma] 
    {data/unsolved_newdetk.csv};
  \end{axis}
\end{tikzpicture}
\caption{\detk}
\end{subfigure}
\hspace{1em}
\begin{subfigure}[b]{0.3\textwidth}
  \centering
  \begin{tikzpicture}[transform shape,scale=0.8]
  \begin{axis}[
    xmin=0, 
    xmode=log,
    ymode=log,
    xlabel={Number of edges},
    yticklabels={,,},
  ]
    \addplot+[color=poscol, only marks,mark size=1pt,mark=+] table[x=Edges, y=Vertices, col sep=comma] 
    {data/solved_htdSMT.csv};
    \addplot +[color=negcol,only marks,mark size=1pt, mark=x] table[x=Edges, y=Vertices, col sep=comma]
    {data/unsolved_htdSMT.csv};
  \end{axis}
\end{tikzpicture}
\caption{\hdsmt}
\end{subfigure}
\hspace{1em}
\begin{subfigure}[b]{0.3\textwidth}
  \centering
  \begin{tikzpicture}[transform shape,scale=0.8]
  \begin{axis}[
    xmin=0, 
    xmode=log,
    ymode=log,
    xlabel={Number of edges},
    yticklabels={,,},
  ]
    \addplot+[color=poscol, mark=+, only marks,mark size=1pt] table[x=Edges, y=Vertices, col sep=comma]
    {data/solved_logk.csv};
    \addplot +[color=negcol, mark=x, only marks,mark size=1pt] table[x=Edges, y=Vertices, col sep=comma]
    {data/unsolved_logk.csv};
  \end{axis}
\end{tikzpicture}
\caption{\logk}
\end{subfigure}

\caption{Comparison of solved instances (green) and unsolved instances (red), relative to their edge and vertex size.}
\label{fig:unsolvedComp}
\end{figure*}
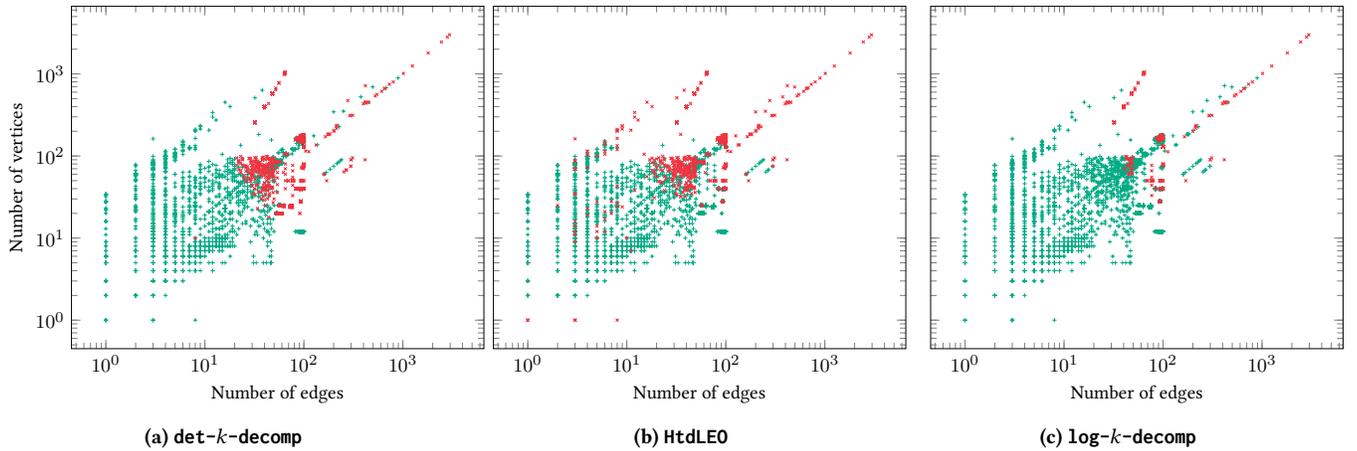

\setlength{\tabcolsep}{2pt}
\begin{table}[t]

  \caption{Comparison of the decomposition methods by how many instances were solved for a specific width. }
  \label{WidthTable}

\newcolumntype{?}{!{\vrule width 1pt}}
\begin{tabular}{c  c ? c c c}
\toprule
Width & Virtual Best & \newdetk & \hdsmt & \logk \\ 
\midrule
1 & 709 & 677 & 649 & \textbf{709} \\ 
2 & 595 & 586 & 567 & \textbf{595} \\ 
3 & 310 & \textbf{310} & 273 & \textbf{310} \\ 
4 & 386 & 379 & 321 & \textbf{386} \\ 
5 & 450 &  38 & 341 & \textbf{450} \\ 
6 & 485 &  28 & 307 & \textbf{480} \\ 
7 & 124 &   9 &  16 & \textbf{108} \\ 
8 & 115 &   1 &  \textbf{69} &  46 \\ 
9 & 19  &   0  &   1 &  \textbf{18} \\
\bottomrule
\end{tabular}

\end{table}

\begin{table*}[t]
  
    \caption{Comparison of the decomposition methods by the upper bounds it could provide. Note that \hdsmt is not being explicitly considered here, since it directly computes the optimal width. Thus it would have the number 2544 --- its number of solved instances -- in each row. }
  \label{Smallwidths}

\newcolumntype{?}{!{\vrule width 1pt}}
\begin{tabular}{c  c ? c c c}
\toprule
Problem to solve & Virtual Best & \logk (Hybrid) & \newdetk & \logk \\ 
\midrule  
$\hw \leq 1$ & 3648 & \textbf{3648} &  \phantom{0}3616\footnotemark & \textbf{3648} \\
$\hw \leq 2$ & 3648 & \textbf{3648} &  3631 & \textbf{3648} \\
$\hw \leq 3$ & 3637 & \textbf{3637} &  3355 & 3567 \\
$\hw \leq 4$ & 3623 & \textbf{3623} & 2391  & 3178 \\
$\hw \leq 5$ & 3616 & \textbf{3611} & 2485  & 2924 \\ 
$\hw \leq 6$ & 3370 & \textbf{3253} & 2897  & 2349 \\

\bottomrule
\end{tabular}

\end{table*}

\begin{table*}
\setlength{\tabcolsep}{4pt}

\newcolumntype{?}{!{\vrule width 1pt}}

    \caption{Extension of Table~\ref{tab:bigeval}, with running times for \hdsmt extended to 10 hours}
    \label{tab:smalleval}
    
\begin{tabular}{c?  c }
\toprule
\begin{tabular}{c c}
\begin{tabular}{r r c}
    \multicolumn{1}{c}{Origin of} & \multicolumn{1}{c}{Size of} &  \multicolumn{1}{c}{Instances in}  \\
    \multicolumn{1}{c}{Instances}  & \multicolumn{1}{c}{Instances} &  \multicolumn{1}{c}{Group} \\
 \midrule

Application & $75< |E| \leq 100$            & 405 \\ 
            & $50< |E| \leq \phantom{0}75 $ & 514 \\   
            & $10< |E| \leq \phantom{0}50 $ & 369 \\ 
            &     $|E| \leq \phantom{0}10 $  & 915 \\ 
\midrule 

Synthetic   & $\phantom{10 < } |E| > 100 $              &  \phantom{0}66 \\  
            &             $75< |E| \leq 100$            & 422 \\ 
            &             $50< |E| \leq \phantom{0}75 $ & 215 \\   
            &             $10< |E| \leq \phantom{0}50 $ & 647 \\ 
            & $\phantom{100<}  |E| \leq \phantom{0}10 $ &  \phantom{0}95 \\ 
\midrule
\multicolumn{1}{c}{Total}   & \multicolumn{1}{c}{-} & 3648 
\end{tabular}

 \end{tabular} &  
 
  \begin{tabular}{r c}

\multicolumn{2}{c}{\hdsmt~\cite{DBLP:conf/ijcai/SchidlerS21} 10 Hour Run} \\  

 \#solved & \multicolumn{1}{c}{Changes 1 hour run}\\ %
\midrule
94 & +\phantom{0}29 \\%

461 & +\phantom{0}13\\ %

237 & $\pm$\phantom{00}0\\ %

876 & $\pm$\phantom{00}0\\ %

 \midrule 

13 & $\pm$\phantom{00}0 \\ %

360 & +\phantom{0}48 \\ %

214 & +\phantom{00}2 \\ %

433 & +130 \\ %

78 & $\pm$\phantom{00}0 \\ %

\midrule

2766 & +222 %

 \end{tabular}

 \\
\bottomrule
\end{tabular}

\end{table*}

\subsection{Analysis on Determining Low Width }

We want to analyse for how many instances of HyperBench, each decomposition method could determine its width and specifically focus on how well it fares as the width increases. Note that \hdsmt is unique in this respect since it determines the optimal width right away. Thus if we ask for how many instances \hdsmt could determine if its width is $\leq 5$, for example, we are really just counting how many timeouts there were in general. For the parametrised decomposition methods, however, this question does give us new insights into how its runtime scales when looking for decompositions of larger or smaller width. 

For this purpose, we first need the concept of the ``Virtual Best'' method. This notion simply aggregates the results of all other methods and shows how for how many instances of HyperBench we know their $\hw$.  We can see in Table~\ref{WidthTable} how each of the three methods fares when compared against this virtual best method. For widths up to 5, the Hybrid \logk is unbeaten, solving all known instances, as well as solving many of them exclusively. 

To provide a more detailed analysis, we also compare for how many instances of $\hw$ up to 6, each method can determine whether an instance has $\hw$ of lower than the given number or not, by finding an HD of such a width or determining that no such HD can exist. Note that this does not require proving optimality.
We can see the results in Table~\ref{Smallwidths}. We can see that both \logk and \logk (Hybrid) are very good at this, with the Hybrid determining for 3253 (or almost 90\% of) instances whether they have $\hw \leq 6$. If we limit ourselves to $\hw \leq 5$, it determines the question for 3611 or almost 99\% of instances.

\footnotetext{We note that the reason \newdetk fails to determine acyclicity (i.e. whether $\hw \leq 1$) for all graphs is due to a bug where it will output a HD of width 2 instead of one of width 1 when given certain acyclic graphs. While of low practical interest, as determining acyclicity is a trivial problem, it is still the case that \newdetk in its current form fails to determine all acyclic graphs of HyperBench correctly.}

\end{document}